\documentclass{llncs}
\usepackage{hyperref}
\newcommand{\norm}[1]{}
\usepackage[utf8]{inputenc}
\usepackage[mathscr]{eucal}
\usepackage{graphicx}
\usepackage{listings}
\usepackage{amssymb}
\usepackage{dsfont}
\usepackage{mathtools}
\usepackage{algorithm}
\usepackage{algpseudocode}
\usepackage{xcolor}
\usepackage{multirow}
\usepackage{amsmath}
\usepackage[english]{babel}
\usepackage{wasysym}
\usepackage{qip}
\usepackage{qm}
\usepackage{csquotes}
\usepackage{environ}

\usepackage[clock]{ifsym}

\usepackage[
	n,
	operators,
	advantage, 
	sets,
	adversary,
	landau,
	probability,
	notions,
	logic,
	ff,
	mm,
        primitives,
        events,
        complexity,
        asymptotics, 
        keys]{cryptocode}

\usepackage[maxbibnames=10]{biblatex}
\usepackage[tikz]{bclogo}

\spnewtheorem{construct}{Construction}{\bfseries}{\itshape}
\spnewtheorem*{thm}{Theorem}{\bfseries}{\itshape}
\spnewtheorem*{lem}{Lemma}{\bfseries}{\itshape}
\newcommand{\hil}[1]{\ensuremath{\mathcal{#1}}}
\newcommand{\hildd}[1]{\ensuremath{\mathscr{#1}}}

\newcommand{\pos}[1]{\ensuremath{\mathrm{Pos}(#1)}}

\newcommand{\idt}{\ensuremath{\mathds{1}}}
\newcommand{\ttt}{\ensuremath{\texttt{t}}}
\newcommand{\ek}{\ensuremath{\textsf{ek}}}
\newcommand{\ck}{\ensuremath{\textsf{ck}}}
\newcommand{\nk}{\ensuremath{\textsf{k}}}
\newcommand{\TLP}{\ensuremath{\textsf{TLP}}}
\newcommand{\OWF}{\ensuremath{\textsf{pq-OWF}}}
\newcommand{\TD}{\ensuremath{\textsf{TD}}}
\newcommand{\TdVK}{\ensuremath{\textsf{TD-DVK}}}
\newcommand{\QROM}{\ensuremath{\textsf{QROM}}}
\newcommand{\QPKE}{\ensuremath{\textsf{QPKE}}}
\newcommand{\BQSM}{\ensuremath{\textsf{BQSM}}}

\newcommand{\Hy}{\ensuremath{\textsf{H}}}

\newcommand{\advs}{\ensuremath{\adv_{\texttt{s}}}}
\newcommand{\ot}{\ensuremath{\mathcal{O}_{\texttt{s}}}}
\newcommand{\s}{\ensuremath{\texttt{s}}}

\renewcommand{\qed}{{\hfill{$\blacksquare$}}}

\begin{document}
\title{How to Sign Quantum Messages}

\author{Mohammed Barhoush \and Louis Salvail}
\institute{Universit\'e de Montr\'eal (DIRO), Montr\'eal, Canada\\
\email{mohammed.barhoush@umontreal.ca}\ \ \ \email{salvail@iro.umontreal.ca}}%

\maketitle

\begin{abstract}
Signing quantum messages has long been considered impossible even under computational assumptions. In this work, we challenge this notion and provide three innovative approaches to sign quantum messages that are the first to ensure authenticity with public verifiability. Our contributions can be summarized as follows:
\begin{enumerate}
    \item \emph{Time-Dependent Signatures:} We introduce the concept of time-dependent ($\TD$) signatures, where the signature of a quantum message depends on the time of signing and the verification process depends on the time of the signature reception. We construct this primitive assuming the existence of post-quantum secure one-way functions ($\OWF$s) and time-lock puzzles ($\TLP$s).

    \item \emph{Dynamic Verification Keys:} By utilizing verification keys that evolve over time, we eliminate the need for $\TLP$s in our construction. This leads to $\TD$ signatures from $\OWF$s with dynamic verification keys. 
    
    \item \emph{Signatures in the BQSM:} We then consider the bounded quantum storage model, where adversaries are limited with respect to their quantum memories. We show that quantum messages can be signed with information-theoretic security in this model.
\end{enumerate}

Moreover, we leverage $\TD$ signatures to achieve the following objectives, relying solely on $\OWF$s:
\begin{itemize}
    \item \emph{Authenticated Quantum Public Keys:} We design a public-key encryption scheme featuring authenticated quantum public-keys that resist adversarial tampering. 
    \item \emph{Public-Key Quantum Money:} We present a novel $\TD$ public-key quantum money scheme.
\end{itemize}
\end{abstract}

\newpage

\section{Introduction}

Given the consistent advancements in quantum computing, it is expected that future communications will feature quantum computers transmitting over quantum channels. A fundamental question naturally arises: how can quantum data be securely transmitted within the emerging quantum internet? One option is for users to share secret keys through secure channels. Yet, this option quickly becomes unwieldy as the number of users grows or when secure channels are not readily available. Another option is to rely on quantum key distribution \cite{BB84}, but this too is inefficient for large-scale applications as it requires several rounds of interaction with each user. In contrast, classical information can be encrypted and authenticated non-interactively via public channels. Can a similar feat be achieved quantumly? 

Towards this goal, several works have shown how to encrypt quantum information from standard assumptions \cite{ABF16,AGM21}. Yet, somewhat surprisingly, signing quantum information has been shown to be impossible \cite{BCG02,AGM21}. On the other hand, classical digital signature schemes have been crucial cryptographic primitives - realizing a range of applications including email certification, online transactions, and software distribution.

As a result of the impossibility, researchers have focused on an alternative approach to quantum authentication called \emph{signcryption}. In this setting, the sender uses the recipient's public encryption key to encrypt a message before signing it. Only the recipient can verify that the signature is authentic by using their secret decryption key which means that signcryption does not allow for public verifiability -- only a single receipient can verify. Such schemes clearly rely on assumptions for public-key encryption such as trapdoor functions. Until this point, it was widely believed that signcryption is the only way to evade the impossibility result. In fact, Alagic, Gagliardoni, and Majenz \cite{AGM21} carried out an in-depth analysis on the possibility of signing quantum information and concluded that \emph{``signcryption provides {the only way} to achieve integrity and authenticity without a pre-shared secret"}. In this work, we address and revisit the questions:

\begin{center}
\emph{Are there alternative methods to authenticate quantum information without a pre-shared secret?}
\end{center}

Interestingly, this question has important implications in quantum public key encryption ($\QPKE$). Traditionally, classical public-key encryption (PKE) can not be built from one-way functions \cite{IR89} and requires stronger assumptions such as trapdoor functions. However, the works \cite{GSV23,C23} show that PKE can be achieved from post-quantum secure classical one-way functions ($\OWF$) if the public-keys are quantum! Yet, these constructions face a serious problem: authenticating a quantum state is difficult. This issue is not addressed in these works; as a result, these constructions need to assume secure quantum channels for key distribution which is quite a strong physical setup assumption given the goal of encryption is to establish secure communication over insecure channels. On the other hand, there are well-established methods to distribute classical keys, for instance through the use of classical digital signatures. Such procedures are referred to as \emph{public-key infrastructure}. 

Ideally, we aim to establish encryption based on public-key infrastructure and one-way functions. More simply, we want to authenticate the quantum keys in $\QPKE$ schemes using classical certification keys. Prior to this work, the only way to do this would be through the use of signcryption. But, then we are back to relying on assumptions of classical PKE; not to mention the inconveniences of having to encrypt each public-key specific to the recipient. In particular, the following is another critical question addressed in this work:

\begin{center}
\emph{Is $\QPKE$ with {publicly-verifiable quantum public-keys} possible from $\OWF$s?}
\end{center}

Another relevant application of quantum signatures pertains to quantum money. Quantum money was first introduced by Weisner \cite{W83} in 1969 as a method to prevent counterfeiting by utilizing the unclonability of quantum states. There are two main categories of quantum money: private-key and public-key. Public-key systems rely on physical banknotes that are hard to counterfeit and are verifiable with a public-key. Meanwhile, private-key systems rely on a third party that oversees transactions and approves or declines payments. The main advantage of public-key systems is that transactions do not require interaction with a third party. As a result, such systems are more feasible for large-scale applications and provide more privacy to users. 

Unfortunately, public-key systems are more difficult to realize. On one hand, several works have constructed private-key quantum money \cite{W83,JLS18,RS19} from standard assumptions. On the other hand, public-key quantum money has only been constructed assuming indistinguishability obfuscation \cite{AC12,Z21,S22} or from new complex mathematical techniques and assumptions \cite{FGH12,KA22,KSS21} which we are only beginning to understand and, some of which, have been shown to be susceptible to cryptoanalytic attacks \cite{BDG23}. Sadly, all existing constructions for indistinguishability obfuscation are built on assumptions that are post-quantum insecure \cite{AJL19,JLS21,JLS22} or on new assumptions \cite{GP21,WW21} that have been shown to be susceptible to cryptoanalytic attacks \cite{HJL21}. 

As a result, public-key quantum money remains an elusive goal. Prompted by this predicament, Ananth, Hu, and Yuen \cite{AHY23} recently conducted a comprehensive investigation into public-key quantum money, culminating in a strong negative result. In particular, they showed that public-key quantum money is impossible to build using any collusion-resistant hash function in a black-box manner under certain conditions. This discussion raises the following important question addressed in this work:

\begin{center}
\emph{Is public-key quantum money possible from standard computational assumptions?}
\end{center}

\subsection{Summary of Results}

The impossibility of signing quantum information was first discussed by Barnum, Crepeau, Gottesman, Smith, and Tapp \cite{BCG02} and later established more rigorously by Alagic, Gagliardoni, and Majenz \cite{AGM21}. Informally, their rationale was that any verification algorithm, denoted as $V$, which can deduce the quantum message (or some information thereof) from a signature, can be effectively inverted $V^\dagger$ to sign a different message.


The central innovation of this work lies in recognizing that, in specific scenarios, this inversion attack consumes a prohibitive amount of resources. This study explores two variations of this concept, with the resource factor taking the form of either time or quantum space, as we clarify in the ensuing discussion. 

The first approach to signing quantum information is to vary the signing and verification procedures with respect to time. We term this approach \emph{time-dependent ($\TD$) signatures} and present two constructions realizing this primitive. The first construction relies on the assumption of \emph{time-lock puzzles} ($\TLP$) \cite{RSW96,BGJ16} and $\OWF$s. Here, the $\TLP$s ensure that the verification procedure demands prohibitive time to inverse. We note that $\TLP$s can be built from $\OWF$s \cite{BGJ16,CFH21} in the quantum random oracle model ($\QROM$) \cite{BD11}, where all parties are given quantum oracle access to a truly random function.

The second construction also enforces a verification process that is time-consuming to invert. However, in this case, this enforcement is achieved more directly by delaying the announcement of the verification key. Specifically, we authenticate and encrypt messages in a symmetric-key fashion and announce the key later. By the time the key is revealed, allowing users to validate old signatures, it is too late to exploit the key for forging new signatures. As a result, the verification key must be continually updated to allow for new signature generation. We denote schemes with dynamic verification keys as \emph{$\TdVK$ signatures} for brevity. An attractive aspect of $\TdVK$ signatures is that they can be based solely on $\OWF$s, yielding surprisingly powerful applications from fundamental assumptions. 

We demonstrate how to utilize $\TdVK$ signatures to build more secure $\QPKE$ schemes where the quantum public-keys are signed with $\TdVK$ signatures. This approach allows basing $\QPKE$ on $\OWF$s and public-key infrastructure, i.e. enabling encryption with tamper-resilient quantum public-keys. Furthermore, we employ $\TdVK$ signatures to construct the first public-key quantum money scheme based on a standard computational assumption namely $\OWF$s. The verification key in this setting is dynamic, preventing a completely offline money scheme. Fortunately, we are able to mitigate this issue and obtain a completely offline public-key quantum money scheme by utilizing our $\TD$ signature scheme, albeit by relying on $\TLP$s.

Our second strategy for signing quantum information involves a verification process that necessitates an impractical amount of quantum memory to invert. To achieve this, we need to assume the adversary's quantum memory (qmemory) size is limited leading us to the \emph{bounded quantum storage model} ($\BQSM$) \cite{dfss05}. We construct a signature scheme for quantum messages in this model that requires users to have $\ell^2$ qubits of quantum memory, where $\ell$ is the size of the quantum message to be signed, and is information-theoretically secure against adversaries with $\s$ quantum memory where $\s$ can be set to any fixed value that is polynomial with respect to the security parameter. Note that $\s$ is not related to $\ell$ and only has an effect on the length of the quantum transmissions involved. 

\subsection{Future Work}

Time acts a form of synchronization between players in a protocol. This work demonstrates the power of utilizing time in cryptography. Specifically, it shows how incorporating time-dependence can significantly aid in the construction of fundamental cryptographic primitives such as $\QPKE$, quantum signatures, and public-key quantum money. This is particularly encouraging given the simplicity and security of implementing time-dependence in practice. Therefore, an interesting avenue for future work is to consider what other cryptographic primitives can benefit from this approach. 

Another possible topic is the relation between the ``time-dependent model'' and alternative models in cryptography. For instance, in the well-studied \emph{common reference string model}, a trusted setup provides all parties with the same string. Time is somewhat related to a form of dynamic common reference string that is completely non-random. A more rigorous abstraction of the time-dependent model might help us understand alternative methods to sign quantum messages. 

Finally, we believe that the $\BQSM$ has not received sufficient attention. Given the practical challenges of storing and operating on quantum states, this model is very well motivated. However, there have been few studies on cryptographic primitives within this model. We hope that this work will encourage further research in this area.

\subsection{Related Work}

To the best of our knowledge, this is the first work to sign quantum messages without using signcryption. As previously discussed, while signcryption is effective in its own right, it lacks the feature of public verifiability and entails a form of interaction, requiring the distribution of public encryption keys to the signer. Consequently, this is the first approach ensuring the authenticity of quantum messages while maintaining public verifiability.

Meanwhile, the topic of public-key quantum money has been extensively studied. Nevertheless, to our knowledge, this is the first work to achieve this highly elusive goal from standard assumptions, albeit with a time dependence factor. However, it is relevant to acknowledge the work of Roberts and Zhandry in the creation of \emph{franchised quantum money} \cite{Z212}. Their franchised approach provides an inventive way to avoid some of the challenges of public-key quantum money by providing users with distinct verification keys. Notably, their scheme is based on $\OWF$s, although, it faces some serious drawbacks. Specifically, it requires secure private channels to each and every user at the start in order to distribute the secret verification keys. Moreover, the scheme only supports a limited number of participants and becomes vulnerable if a sufficiently large group of users were to collude.

Finally, no prior works presented methods to authenticate quantum public-keys in $\QPKE$. However, concurrently and independently of our work, Kitagawa, Morimae, Nishimaki, and Yamakawa \cite{KMN23} explored the concept of $\QPKE$ with tamper-resilient quantum public-keys, achieving a similar goal. They introduced a security notion akin to IND-CKPA and build a scheme satisfying this notion through a different approach without relying on time. Specifically, they rely on recently discovered properties of quantum states \cite{AAS20} to create tamper-resilient quantum public keys. While their scheme is CKPA-secure under the assumption of $\OWF$s, it does not allow users to verify or detect modified keys, potentially leading to invalid encryptions that exhibit decryption errors. Their work concludes with two open questions: 
\begin{enumerate}
    \item Is it possible to design a CKPA-secure $\QPKE$ scheme without decryption errors?
    \item Can a CKPA-secure $\QPKE$ scheme be built based on potentially weaker assumptions than $\OWF$s?
\end{enumerate}

Both questions are answered affirmatively in our work by incorporating time-dependence, as detailed in Sec.~\ref{sec:QPKE}. In Constructions \ref{con:Asy} \& \ref{con:pure-state}, keys are publicly verifiable and thus ciphertexts exhibit no decryption errors, addressing the first question. Furthermore, Construction \ref{con:pure-state} is CKPA-secure and is based on PRFSs, addressing the second question. Our approach also has the advantage of being more generally applicable to any quantum public-key scheme and achieving a stronger level of security, demonstrating the effectiveness of time as a cryptographic tool.    

{Another related notion is signatures with certified deletion \cite{MPY23}. In this protocol, the receiver of a signature can perform an operation to delete the signature, producing a certificate that verifies the deletion. In our case, the signature is invalidated after a certain period regardless of the receiver's actions. Note that the  techniques of certified deletion cannot be used to sign a quantum message, as otherwise, it would imply a standard signature scheme of quantum messages. 
}
\section{Technical Overview}

We now describe each contribution in more detail. 

\paragraph{Time-Dependent Signatures.}

A $\TD$ signature scheme is the same as a standard signature scheme except that the signing and verification algorithms depend on an additional variable: \emph{time}. Henceforth, we let $\ttt$ denote the current time and it is assumed that all parties are aware of the current time up to some small error. {This may implicitly assume a trusted third party for clock synchronization.} More formally, $\TD$ signatures are defined as follows:
{
\begin{definition}[Time-Dependent Signatures]
A \emph{$\TD$ signature scheme} $\Pi$ over quantum message space $\hildd{Q}$ consists of the following algorithms: 
\begin{itemize}
    \item $\textsf{KeyGen}(1^\lambda)$: Outputs $(\sk,\vk)$, where $\sk$ is a secret key and $\vk$ is a verification key.
    \item $\textsf{Sign}(\sk,\ttt, \phi):$ Outputs a signature $\sigma$ for $\phi \in \hildd{Q}$ based on the time $\ttt$ using $\sk$. 
    \item $\textsf{Verify}(\vk, \ttt, \sigma')$: Verifies the signature ${\sigma}'$ using $\vk$ and time $\ttt$. Correspondingly outputs a message $\phi'$ or $\perp$. 
\end{itemize}
\end{definition}
}
We now describe the construction for a $\TD$ signature in based on $\TLP$s and $\OWF$s.

$\TLP$s is a cryptographic primitive that enables hiding a message for a time $t$, but allows for decryption in a similar time $t'\approx t$. 
{
\begin{definition}[Time-Lock Puzzles (Informal)]
    A $\TLP$ consists of two algorithms: 
\begin{enumerate}
        \item $\textsf{Puzzle.Gen}(1^\lambda,t, s):$ Generates a classical puzzle $Z$ encrypting the solution $s\in \{0,1\}^\lambda$ with difficulty parameter $t$.
        \item $\textsf{Puzzle.Sol}(Z):$ Outputs a solution $s$ for the puzzle $Z$. 
    \end{enumerate}
    Security informally requires that no QPT adversary can decipher the  solution $s$ of a puzzle $Z$ in time $t$.
\end{definition} }
    
Bitansky, Goldwasser, Jain, Paneth, Vaikuntanathan, and Waters \cite{BGJ16} showed that $\TLP$ can be constructed, assuming the existence of $\OWF$s, from any non-parallelizing language. This result holds true even in the quantum realm. It is worth noting that non-parallelizing languages have been derived from a random oracle \cite{CFH21}. 

We provide a brief description of our construction for $\TD$ signatures. We denote the algorithms for a one-time symmetric authenticated encryption scheme on quantum messages as $({\textsf{1QGen}},{\textsf{1QEnc}},{\textsf{1QDec}})$, which exists unconditionally \cite{BCG02}. {Such a protocol allows for the secure authenticated encryption of a single quantum message.} To sign a quantum message $\phi$, we sample a key $\nk\leftarrow {\textsf{1QGen}}(1^\lambda)$ and authenticate the message as $\tau\leftarrow {\textsf{1QEnc}}(\nk, \phi)$. Following this, we generate a $\TLP$ $Z$ requiring 1 hour to solve and whose solution is the message $(\nk,T, \textsf{sig})$, where $T$ corresponds to the time at the moment of signing and $\textsf{sig}$ is a signature of $(\nk,T)$ under a classical signature scheme. Note that signatures on classical messages exist assuming $\OWF$s \cite{R90}\footnote{This reduction is classical but holds in the quantum setting as well.} Consequently, the signature of $\phi$ is $(\tau, Z)$. 

Assume that a receiver gets the signature at time $T'$. To verify, the receiver solves the puzzle $Z$ to learn $(\nk,T, \textsf{sig})$. If the signature $\textsf{sig}$ is valid and the time of reception is close to $T$, i.e. $T'$ is within half an hour from $T$, then the verifier outputs ${\textsf{1QDec}}(\nk,\tau)$ to retrieve $\phi$. However, it is crucial to understand that the verifier can no longer use the pair $(\nk,T)$ to sign a new message because by the time the puzzle is unlocked, it has already become obsolete! Specifically, the time elapsed is at least $\ttt\geq T+1$, leading future verifiers to reject a message signed under time $T$.  

Note that an alternative approach to using $\TLP$s would be to simply send the encoded message $(\nk,T, \textsf{sig})$ after some delay. Hence, a signature would then consist of two parts. We do not consider this solution further as it follows in the same way, albeit we need to use delays in the signature transmission. Later, we discuss a solution that does not rely on $\TLP$s or on delays.

Notice how the use of $\TLP$ (or delays) gives us the ability to add intricacy to the verification process, and this is precisely what is needed to circumvent the impossibility result. 

\paragraph{The Issue with $\TLP$s.}

Our construction of $\TD$ signatures is seemingly hindered by its reliance on $\TLP$. Formally, a $\TLP$ does not actually guarantee that no algorithm can decrypt without sufficient time; rather, it ensures that no algorithm can decrypt without sufficient algorithmic depth (CPU-time). We would like to translate this CPU-time into real-time. However, we have limited control over this conversion due to the variations in computer speeds. For instance, silicon gates are slower than gallium arsenide gates. 

On a positive note, exact equivalence between CPU-time and real-time is not necessary. As long as there exists a bound on CPU-time based on real-time, then using this bound we can ensure security of our $\TD$ signatures. Explicitly, we modify the verification process to only pass if the time of reception is close to the time of generation, but the $\TLP$ requires much longer to solve depending on this bound. This prevents even a fast computer from unlocking the puzzle before it must send it, while still allowing slower computers to validate signatures, albeit with some delay. 

However, such practical investigations are out of the scope of this work but can be explored more rigorously in future work. In this work, we only focus on the theoretical aspect and assume that CPU-time and real-time are equivalent.

\paragraph{Time-Dependent Signatures with Dynamic Verification Keys.}

We now show how we can discard $\TLP$s, giving $\TD$ signatures in the plain model assuming only $\OWF$s. In this construction, we introduce a novel element: dynamic verification keys that evolve with time. First, note that post-quantum pseudorandom functions (\textsf{pq-PRF}) and symmetric key quantum authenticated encryption scheme (\textsf{QAE}) on quantum messages have both been derived from $\OWF$s \cite{Z12,AGM18}. Let $({\textsf{QEnc}},{\textsf{QDec}})$ denote the algorithms for such an \textsf{QAE} scheme and let $F$ be a \textsf{pq-PRF}.

Here is how our construction unfolds: we choose a time-independent secret key $\sk$ and generate temporary signing keys $\nk_i\coloneqq F(\sk,i)$ to be used within the time interval $[t_i,t_{i+1})$, {where $t_i\coloneqq i\cdot \delta$ for some fixed time $\delta>0$}. To sign a message $\phi$ within this time frame, we simply output $\tau_i \leftarrow {\textsf{QEnc}}(\nk_i, \phi)$. At time $t_{i+1}$, we announce the verification key $\vk_i\coloneqq (t_i,t_{i+1},\nk_i)$. This allows any recipient that received $\tau_i$ to decrypt and verify the message. Users are required to take note of the time they receive a signature and store it until the verification key is unveiled. 

Importantly, an adversary cannot predict a signing key $\nk_i$ by the security of \textsf{pq-PRF}s. Consequently, they cannot forge a ciphertext in the time frame $[t_i,t_{i+1})$ by the security of the \textsf{QAE} scheme. Later, at time $t_{i+1}$, any attempt to forge a signature using $\nk_i$ becomes invalid since users will then verify using newer keys.


On the downside, $\TD$ signatures come with an expiration period, meaning they must be utilized within a relatively short time frame from their generation. Note that the same issue exists in our first construction as well. This limitation may restrict their utility in certain applications compared to traditional signatures. However, in some scenarios, this limitation can also be seen as an advantage. In particular, our signatures satisfy a form of ``disappearing'' security which was heavily studied in the bounded memory (classical and quantum) models. Specifically, in our schemes, an adversary is incapable of forging a signature of a message $\phi$ even if it received a signature of $\phi$ earlier!

\paragraph{Signatures in the Bounded Quantum Storage Model.}

Next, we introduce an alternative approach to sign quantum information within the $\BQSM$ without relying on time. In this model, an adversary $\advs$ possesses unlimited resources at all times except at specific points, where we say the \emph{qmemory bound applies}, reducing $\advs$'s quantum memory to a restricted size of $\texttt{s}$ qubits\footnote{{ Note that in the $\BQSM$, schemes rely on two parameters: adversary's memory $\texttt{s}$ and the security parameter $\lambda$. The two parameters are not necessarily dependent, but in our schemes they are because the length of transmissions depends on $\texttt{s}$ and correctness requires this to be polynomial in $\lambda$.}}. The adversary is again unrestricted after the memory bound applies and is never restricted with respect to its computational power. 

Our construction is inspired by the \textsf{BQS} signature scheme for classical messages presented by Barhoush and Salvail \cite{BS23}. We first briefly describe their construction. In their scheme, the public-key is a large quantum state and each user can only store a subset of the quantum key. Consequently, different users end up with distinct subkeys. This idea is realized through a primitive called \emph{program broadcast} which was achieved unconditionally in the $\BQSM$ \cite{BS23} and is defined in Sec.~\ref{sec:def bqs}\footnote{In essence, a $\nk$-time program broadcast of a function $P$ permits a polynomial number of users to each obtain a single evaluation of $P$ on their chosen inputs while simultaneously preventing any single user from learning more than $\nk$ evaluations of $P$.}. The secret key is an appropriate large function $P_{\pk}$ and the public-key consists of a quantum program broadcast of $P_{\pk}$. Each user acquires a distinct evaluation of $P_{\pk}$ which serves as a \emph{verification subkey}. The broadcast is terminated after all users have received key copies and before the signing process.  

To facilitate signing, which must be verifiable by all users, one-time programs are utilized -- this powerful primitive can be built unconditionally in the $\BQSM$ \cite{BS23}\footnote{Informally, a one-time program can be used to learn a single evaluation of a program without revealing any additional information about the program.}. We provide a simplified description as the actual construction is somewhat technical. 

Split the output of $P_{\pk}(x)=(P_1(x),P_2(x))$ in half. The signature of a message $m$ is a one-time program that on an input of the form $(x,P_1(x))$, for some $x$, outputs $(P_2(x),m)$. Consider an adversary $\adv$ that wants to forge a signature to deceive a user, lets call them Bob, who possesses $(x_B,P_{\pk}(x_B))$ from the broadcast. Bob would evaluate the forged program on $(x_B,P_1(x_B))$ but there is a very low probability the output will contain $P_2(x_B)$ as $\adv$ is unlikely to have learned $P_{\pk}(x_B)$ from the broadcast. 

The problem is that $\adv$ can take a valid signature, edit it, and forward it to Bob. Even though $\adv$ is limited with respect to its qmemory and cannot store signatures, this attack can be performed on the fly. One-time programs are potent tools due to their capacity to conceal information, but this attribute also renders it challenging for Bob to detect if a received signature has been tampered with. Exploiting this vulnerability, an adversary can modify a signature of $m$ to sign a new message $m'$. Unfortunately, Barhoush and Salvail's work \cite{BS23} did not resolve this issue, which consequently restricts their security to scenarios where the adversary is not allowed to query the signature oracle when it attempts to forge a signature and can only query it prior. Practically speaking, there is no way to enforce such a requirement. 

Another issue with their construction is that the public-keys are quantum and yet are not certified. Their work assumes the keys can be distributed authentically to each and every user. Note that this issue cannot be solved by applying $\TD$ signatures to sign the keys because $\TD$ signatures can only be verified after some delay and users in the $\BQSM$ cannot store the quantum signatures. To sum up, there are 3 issues with their construction which we address in this work:

\begin{enumerate}
    \item Can only be used to sign classical messages.
    \item Security is only ensured against adversaries with restricted access to the signing oracle. 
    \item Assumes the secure distribution of quantum public-keys. 
\end{enumerate}

We developed a solution that resolves these issues and achieves quantum unforgeability with certified quantum public-keys against adversaries with unrestricted access to the signing oracle. Our construction is quite technical so we refer the reader to Sec.~\ref{sec:bqsm sig} for more details.

\paragraph{Application: Public-Key Quantum Money.}

We present two public-key quantum money schemes which are derived as applications of $\TD$ and $\TdVK$ signature. In particular, the banknote generation and verification algorithms depend on time so we term our schemes \emph{$\TD$ public-key quantum money}. The first scheme we describe, derived from $\TdVK$ signatures, accomplishes semi-online public-key quantum money assuming $\OWF$s. Meanwhile, the second scheme we describe, derived from $\TD$ signatures, achieves offline public-key quantum money assuming $\OWF$s and $\TLP$s.    

We now provide a simplified overview of our first scheme. We let $({\textsf{1QEnc}},{\textsf{1QDec}})$ denote a one-time \textsf{QAE} scheme and let $F$ denote a \textsf{pq-PRF}. For any $i\in \mathbb{N}$, the bank generates the key $\nk_i\coloneqq F(\sk,i)$ using a secret key $\sk$ and the \textsf{pq-PRF}. To create a quantum banknote at time in $[t_j,t_{j+1})$, the bank samples a random string $y$, and generates a set of $n$ keys $y_i\coloneqq F(\nk_i,y)$ for $i\in [j,j+n]$, {where $n$ is a parameter that can be set by the bank depending on an efficiency/lifespan tradeoff.} Subsequently, the bank prepares the quantum state $\ket{0}$ and, for each key $(y_i)_{i\in [j:j+n]}$, encrypts the state successively using the one-time encryption scheme. The resulting banknote is of the form $\ket{\$}  \coloneqq (y, {\textsf{1QEnc}}_{y_{j}}( {\textsf{1QEnc}}_{y_{j+1}}(...{\textsf{1QEnc}}_{y_{j+n}}(\ket{0})...)$.

Assume Alice receives the banknote at a time in $[t_j,t_{j+1})$. At time $t_{j+1}$, the bank publicly announces $\nk_j$. Alice computes $y_j=F(\nk_j,y)$ and validates the banknote by applying ${\textsf{1QDec}}_{y_j}$. More generally, any user can validate the banknote as long as it is received before time $t_{j+n}$ by performing verification with the appropriate key. However, after time $t_{j+n}$, all the keys used in generating $\ket{\$}$ would have been revealed rendering the banknote unverifiable. Hence, it is imperative to return the banknote to the bank before this expiration. 

We term our scheme ``semi-online'' as users do not need to interact with a bank as in online schemes but they must receive key announcements to validate transactions which prevents complete isolation as in the offline setting. Importantly, our cash system is designed to be scalable and applicable on a macro-level, as it does not rely on a third party to validate each and every transaction as in private-key and online schemes. Nevertheless, a scheme allowing users to operate in a completely offline manner provides further advantages.

Fortunately, our second approach achieves this goal. However, security is only guaranteed assuming $\TLP$s and $\OWF$s. The construction is very similar to the one just described, except the keys $y_j,y_{j+1},...,y_{j+n}$ are encoded in a $\TLP$, allowing users to deduce $y_i$ only at time $t_i$. Here, the $\TLP$ effectively performs the same function as the dynamic verification keys. The puzzle is signed classically ensuring authenticity with a time-independent verification key. Thus, this approach provides a completely offline public-key scheme. 

Our money schemes have two drawbacks: the banknotes expire and the transactions take up to $\delta\coloneqq t_{i+1}-t_i $ time to validate. The first issue can be mitigated by increasing the number of authentication layers in a banknote (choosing a larger value for $n$). However, this would require users to allocate more quantum memory. The second issue can be addressed by announcing keys more frequently (decreasing $\delta$) but this would shorten the lifespan of the banknotes. Ideally, the bank should find a compromise between these two issues. Note that traditional money also has an expiration date and large transactions, even with cryptocurrencies, need to be validated with proper time-consuming procedures. Therefore, these issues are inherent in current financial systems. 

To sum up, we believe the drawbacks of our construction are a reasonable compromise to achieve a publicly-verifiable currency that resists counterfeiting from standard assumptions. To our knowledge, this is the first work to achieve this challenging goal. 

\paragraph{Application: Encryption with Authenticated Public Keys.}

Our final contribution is to present a $\QPKE$ scheme that incorporates publicly verifiable quantum public-keys. Traditionally, $\QPKE$ schemes require that the quantum keys are distributed securely to all user \cite{BS23,GSV23,C23}, making them susceptible to attacks that interfere with the key distribution process. We introduce a security notion that we term \emph{indistinguishability under adaptive chosen quantum key and ciphertext attack} (IND-qCKCA2). In this setting, an adversary can tamper with the quantum public-keys, but it is assumed that classical certification keys can be distributed authentically. More explicitly, the adversary can \emph{choose} the quantum key used by the experiment for encryption but cannot tamper with the classical certification keys, reflecting a realistic threat scenario where classical information can be reliably announced through public-key infrastructure, but quantum information cannot. Additionally, the adversary has quantum access to the encryption and decryption oracles (except on the challenge ciphertext) as in the standard IND-qCCA2 experiment \cite{BZ133,GSM20}. 

To achieve qCKCA2-security, we essentially sign the public-keys using $\TD$ signatures. Surprisingly, the adversary may still be able to mount quantum attacks by storing a state entangled with a user's public-key copy. This stored state can be exploited to extract information on the user's encryptions. To deal with this attack, we encode quantum keys in BB84 states that satisfy an unclonability property known as the monogamy of entanglement \cite{TFKW13}, preventing such splitting attacks. 

We apply $\TdVK$ signatures to the public-keys in two ways to give two constructions. The first construction relies on many-time $\TdVK$ signatures based on $\OWF$s. Since the verification key evolves with time, the resulting signed keys also change and are thus mixed-state. Note that other works also utilize mixed-state quantum keys \cite{BS23,KMN23}. We apply this approach to a $\QPKE$ scheme that is inspired by the construction given in \cite{GSV23} to get the first qCKCA2-secure $\QPKE$ scheme from $\OWF$s. 

Next, by constraining the distribution of the quantum keys to some limited time frame, we show how to upgrade a $\QPKE$ scheme to generate publicly verifiable pure-state quantum keys without requiring any additional computational assumptions. The idea is to encrypt all pure-state keys with a one-time authenticated symmetric key encryption \cite{BCG02}. Then, after all users have received their encrypted key copies, the secret encryption key is revealed, enabling users to validate their keys. After the key is revealed, the distribution terminates, thus characterizing it as a limited-time key distribution. As a direct application to the scheme in \cite{GSV23}, we achieve $\QPKE$ satisfying indistinguishability under chosen key and plaintext attack (IND-CKPA) with limited-time pure-state public-keys from pseudorandom function states (PRFS) with super-logarithmic input size. This is noteworthy because PRFS is potentially a weaker assumption than $\OWF$. More specifically, $\OWF$s imply PRFSs with super-logarithmic input-size \cite{AQY22} and there exists an oracle separation between the two assumptions \cite{K21}. See Sec.~\ref{sec:QPKE} for more details. 

\section{Preliminaries}
\subsection{Notation}
We denote Hilbert spaces by calligraphic letters, such as $\hildd{H}$, and the set of positive semi-definite operators in $\hildd{H}$ as $\pos{\hildd{H}}$. If a Hilbert space $\hildd{X}$ holds a classical value, then we denote this variable as $\text{X}$.

The density matrix of a quantum state in a register $E$ is denoted as $\rho_E$ and, generally, lowercase Greek letters such as $\sigma, \phi, \tau$ denote density matrices of quantum states. The trace norm is denoted as $\|\rho\|\coloneqq \trace{\sqrt{\rho \rho^\dagger}}$, the trace distance as $\delta(\rho,\sigma)\coloneqq \frac{1}{2}\| \rho-\sigma\|$, and the fidelity by $\delta_F$. We say two states in registers $A$ and $B$, respectively, are \emph{$\epsilon$-close} denoted as $A\approx_{\epsilon}B$ if $\delta(\rho_A,\rho_B)\leq \epsilon$ and are \emph{$\epsilon$-independent} if $\rho_{AB}\approx_{\epsilon}\rho_{A}\otimes \rho_{B}$. Note that two states are negligibly close if and only if their fidelity is negligibly close to 1 \cite{FV99}. We let $\langle \phi, \rho\rangle$ denote a sequential transmission where $\phi$ is first sent, and then $\rho$ is sent after. 

For the complex space $\mathbb{C}^2$, the computational basis $\{|0\rangle,|1\rangle\}$ is denoted as $+$, while the diagonal basis $\{|0\rangle_{\times},|1\rangle_{\times}\}$ is denoted as $\times$ where $|0\rangle_{\times}\coloneqq \frac{|0\rangle +|1\rangle}{\sqrt{2}}$ and $|1\rangle_{\times}\coloneqq \frac{|0\rangle -|1\rangle}{\sqrt{2}}$. The pair of basis $\{+,\times\}$ is often called the BB84 basis \cite{BB84}. 

We say $x\leftarrow X$ if $x$ is chosen from the values in $X$ according to the distribution $X$. If $X$ is a set, then $x\leftarrow X$ denotes an element chosen uniformly at random from the set. We say $A$ is \emph{QPT} if it is a quantum polynomial-time algorithm. We let $A^{qP}$ denote an algorithm with access to $q$ queries to an oracle of $P$, $A^P$ denote access to a polynomial number of oracle queries, and $A^{\ket{P}}$ denote access to a polynomial number of quantum oracle queries. For any probabilistic algorithm $G(x)$ which makes $p$ coin tosses, we let $G(x;y)$ be the deterministic version of $G$ where the coin toss outcomes are determined by $y\in \{0,1\}^p$.   

Also, we let $[n]\coloneqq [0,1,...,n-1]$ and $\textsf{negl}$ denote any function that is asymptotically smaller than the inverse of any polynomial. For simplicity, we let $\{0,1\}^{m,n}$ denote the set $\{0,1\}^m\times \{0,1\}^n$. By abuse of notation, given a string $x$ and matrix $M$, then, for simplicity, we let $M\cdot x$ denote the string representation of the vector $M\cdot \vec{x}$.

\subsection{Algebra}
The following algebraic result from \cite{BS23} states that it is difficult to predict the output of a large matrix given a few evaluations. 

\begin{lemma}[\cite{BS23}]
\label{Raz 2}
Let $M$ be an arbitrary $\ell \times n$ binary matrix. Let $A$ be an algorithm that is given as input: $(a_1,b_1),...,(a_{m},b_{m})$ and $(\hat{a}_1,\hat{b}_1),...,(\hat{a}_{p},\hat{b}_{p})$  where $a_i, \hat{a}_i \in \{0,1\}^n$, $b_i,\hat{b}_i\in \{0,1\}^{\ell}$, $m<n$, $p$ is a polynomial in $\ell$, $b_i=M\cdot a_i$ and $\hat{b}_i\neq M\cdot \hat{a}_i$. Then the following statements hold:
\begin{enumerate}
    \item For any vector $a\in \{0,1\}^n$ not in the span of $\{a_1,...,a_m\}$, if $A$ outputs a guess $b'$ of $b\coloneqq M\cdot a$, then $\Pr{[b'=b]}\leq O(2^{-\ell})$.
    
    \item Let $x_0,x_1\in \{0,1\}^n$ be any two distinct vectors not in the span of $\{a_1,...,a_m\}$. Choose $r\leftarrow \{0,1\}$. If $A$ is additionally given $x_0,x_1, y_r$ where $y_r\coloneqq M\cdot x_r$ and outputs a guess $r'$, then $\Pr{[r'=r]}\leq \frac{1}{2}+ O(2^{-\ell})$.
\end{enumerate}
\end{lemma}

\subsection{Min-Entropy}

We recall the notion of min-entropy of a classical variable conditioned on a quantum state. The reader is referred to \cite{KRS09} for more detailed definitions.   

\begin{definition}[Conditional Min-Entropy]
\label{def:cond entropy}
Let $\rho_{XB}\in \pos{\hil{H}_X\otimes \hil{H}_B}$ be classical on $\hil{H}_X$. The min-entropy of $\rho_{XB}$ given $\hil{H}_B$ is given by
\begin{align*}
    H_{\infty}(X|B)_{\rho} \coloneqq - \lg p_{\text{guess}}(X|B)_\rho
\end{align*}
where $p_{\text{guess}}(X|B)_\rho$ is the maximal probability to decode $X$ from $B$ with a POVM on $\hil{H}_B$.  
\end{definition}

\begin{definition}[Smooth Min-Entropy]
Let $\epsilon \geq 0$ and $\rho_{XB}\in \pos{\hil{H}_X\otimes \hil{H}_B}$ be classical on $\hil{H}_X$. The $\epsilon$-smooth min-entropy of $\rho_{XB}$ given $\hil{H}_B$ is given as 
\begin{align*}
    H_{\infty}^{\epsilon}(X|B)_{\rho}\coloneqq \sup_{\overline{\rho}}H_{\infty}(X|B)_{\overline{\rho}}
\end{align*}
where the supremum is taken over all density operators $\overline{\rho}_{XB}$ acting on $\hil{H}_X\otimes \hil{H}_B$ such that $\delta(\overline{\rho}_{XB},\rho_{XB})\leq \epsilon$. 
\end{definition}

The min-entropy satisfies the following chain rule.

\begin{lemma}[Chain Rule for Min-Entropy \cite{RK04}]
\label{chain}
Let $\epsilon\geq 0$ and $\rho_{XUW}\in \pos{\hildd{H}_X\otimes \hildd{H}_U\otimes \hildd{H}_W}$ where register $W$ has size $n$. Then,
\begin{align*}
   H_{\infty}^{\epsilon}({X}|UW)_{\rho} \geq  H^{\epsilon}_{\infty}(XE|W)_{\rho}-n .
    \end{align*}
\end{lemma}

\subsection{Privacy Amplification}
We state the standard approach to amplify privacy through the use of hash functions. 

\begin{definition}
    A class of hash functions $\textsf{H}_{m,\ell}$ from $\{0,1\}^m$ to $\{0,1\}^{\ell}$ is \emph{two-universal} if for $F \leftarrow  \textsf{H}_{m,\ell}$, 
    \begin{align*}
        Pr[F(x)=F(x')]\leq \frac{1}{2^{\ell}}.
    \end{align*} 
\end{definition}

\begin{theorem}[Privacy Amplification \cite{RK04}]
\label{privacy amplification}
Let $\epsilon \geq 0$ and $\rho_{XB}\in \pos{\hildd{H}_X\otimes \hildd{H}_{B}}$ be a cq-state 
where $X$ is classical and takes values in $\{0,1\}^m$. Let $\textsf{H}_{m,\ell}$ be a class of two-universal hash functions and let $F\leftarrow  \textsf{H}_{m,\ell}$. Then, 
\begin{align*}
    \delta(\rho_{F(X)FB},\mathds{1}\otimes \rho_{FB})\leq \frac{1}{2}2^{-\frac{1}{2}(H^{\epsilon}_{\infty}({X}|B)_{\rho}-\ell)}+\epsilon 
\end{align*}
where $\idt$ is uniformly distributed. 
\end{theorem}

\subsection{Monogamy of Entanglement}
\label{sec:mon}
We recall the monogamy of entanglement game which provides a way to test the unclonability of quantum states. This will be useful in proving security of our $\QPKE$ scheme. There are different variants of this game but we recall a version similar to the one presented in \cite{TFKW13}.  

\smallskip \noindent\fbox{%
    \parbox{\textwidth}{%
\textbf{Experiment} $\textsf{G}_{\textsf{BB84}}(n)$:
\begin{enumerate}
\item Alice chooses $x\leftarrow \{0,1\}^n$ and $\theta \leftarrow  \{+,\times\}^{n}$.
\item Alice sends the state $\ket{x}_\theta\coloneqq \ket{x_0}_{\theta_0}\ket{x_1}_{\theta_1}...\ket{x_{n-1}}_{\theta_{n-1}}$ to Bob and Charlie. 
\item Bob and Charlie split the state between them. From this point on they are not allowed to communicate.  
\item Alice sends $\theta$ to Bob and Charlie.
\item Bob and Charlie output $x_B$ and $x_C$, respectively. 
\item The output of the experiment is 1 if $x=x_B=x_C$ and 0 otherwise. 
\end{enumerate}}}
\smallskip

\begin{theorem}[Monogamy of Entanglement, Theorem 3 in \cite{TFKW13}]
\label{thm:mono}
For any quantum algorithms Bob and Charlie,
\label{monogamy}
\begin{align*}
    \Pr{[\textsf{G}_{\textsf{BB84}}(n)=1]}\leq 2^{-n/5} .
\end{align*}
\end{theorem}

\subsection{Pseudorandom Notions}
In this section, we recall the definitions of post-quantum pseudorandom functions (\textsf{pq-PRF}s) and quantum-accessible pseudorandom function state generators (PRFSs). 

\begin{definition}[Post-Quantum Secure Pseudorandom Function]
    Let $PRF=\{PRF_n\}_{n\in \mathbb{N}}$ where $PRF_n:\mathcal{K}_n\times \mathcal{X}_n\rightarrow \mathcal{Y}_n$ is an efficiently computable function. Then, $PRF$ is a post-quantum secure pseudorandom function if for any QPT algorithm $\adv$ there exists negligible function $\textsf{negl}$ such that for any $n\in \mathbb{N}$:
    \begin{align*}
        \lvert \Pr_{k\leftarrow \mathcal{K}_n}[\adv^{PRF_n(\nk,\cdot)}(1^n)=1]-\Pr_{O\leftarrow \mathcal{Y}_n^{\mathcal{X}_n}}[\adv^{O}(1^n)=1]\rvert \leq \negl[n].
    \end{align*}
    where $\mathcal{Y}_n^{\mathcal{X}_n}$ is a function assigning each value in $\mathcal{X}_n$ a random value in $\mathcal{Y}_n$.
\end{definition}

Next, PRFS were first introduced in \cite{AQY22} and the follow-up work \cite{AGQ22} provided a stronger definition that allowed the adversary quantum access to the oracle. We now recall their definition here: 

\begin{definition}[Quantum-Accessible PRFS Generator]
    A QPT algorithm $G$ is a quantum-accessible secure pseudorandom function state generator if for all QPT (non-uniform) distinguishers $\adv$ if there exists negligible function $\textsf{negl}$ such that for any $n\in \mathbb{N}$:
    \begin{align*}
        \lvert \Pr_{k\leftarrow \{0,1\}^{\lambda}}[\adv^{\ket{O_{PRF}(\nk,\cdot)}}(1^n)=1]-\Pr[\adv^{\ket{O_{\textit{Haar}}(\cdot)}}(1^n)=1]\rvert \leq \negl[n].
    \end{align*}
    where, 
    \begin{enumerate}
        \item $O_{PRF}(\nk,\cdot)$, on input a $d$-qubit register $\textit{X}$ does as follows: applies a channel that controlled on register $\textit{X}$ containing $x$, it creates and stores $G_n(\nk,x)$ on a new register $\textit{Y}$ and outputs the state on register $\textit{X}$ and $\textit{Y}$.
        \item $O_{\textit{Haar}}(\cdot)$, modeled as a channel, on input a $d$-qubit register $\textit{X}$, does as follows: applies a channel that, controlled on the register $\textit{X}$ containing $x$, stores $\proj{\theta_x}$ in a new register $\textit{Y}$ where $\ket{\theta}_x$ is sampled from the Haar distribution and outputs the state on register $\textit{X}$ and $\textit{Y}$.
    \end{enumerate}
\end{definition}

\subsection{Signatures on Classical Messages}

We recall the definition of a digital signature scheme on classical messages. 

\begin{definition}[Digital Signatures]
A \emph{digital signature} scheme over classical message space $\hildd{M}$ consists of the following QPT algorithms: 
\begin{itemize}
    \item $\textsf{Gen}(1^\lambda)$: Outputs a secret key $\sk$ and a verification key $\vk$.
    \item $\textsf{Sign}(\sk,\mu):$ Outputs a signature ${\sigma}$ for $\mu \in \hildd{M}$ using $\sk$. 
    \item $\textsf{Verify}(\vk,\mu', {\sigma'})$: Verifies whether ${\sigma'}$ is a valid signature for $\mu' \in \hildd{M}$ using $\vk$ and correspondingly outputs $\top/\perp$.
\end{itemize}
\end{definition}

\begin{definition}[Correctness]
A digital signature scheme is \emph{correct} if for any message $\mu\in \hildd{M}$:
\begin{align*} \Pr{\left[
\begin{tabular}{c|c}
 \multirow{2}{*}{$\textsf{Verify}(\vk,\mu, {\sigma})=\top\ $} &   $(\sk,\vk)\ \leftarrow \textsf{Gen}(1^\lambda)$ \\ 
&  ${\sigma}\ \leftarrow \textsf{Sign}(\sk,\mu)$\\
 \end{tabular}\right]} \geq 1-\negl[\lambda] .
\end{align*}
\end{definition}

We recall the security experiment testing \emph{existential unforgeability} (UF) of a digital signature scheme. 

\smallskip \noindent\fbox{%
    \parbox{\textwidth}{%
\textbf{Experiment} $\textsf{Sign}^{\text{UF}}_{\Pi, \adv}({\lambda})$:
\begin{enumerate}
    \item Sample $(\sk,\vk)\leftarrow \textsf{Gen}(1^\lambda)$.
    \item $\adv$ is given $\vk$ and classical access to the signing oracle $\textsf{Sign}(\sk, \cdot)$. Let $\mathcal{Q}$ denote the set of messages queried.
    \item $\adv$ outputs $(m^*,\sigma^*).$
   \item The output of the experiment is $1$ if
   \[
       m^*\notin \mathcal{Q} \text{ and } \textsf{Verify}(\vk, m^*, \sigma^*)=1 .
   \]
   \item Otherwise, the output is 0. 
\end{enumerate}}}
\smallskip

\begin{definition}
    A digital signature scheme $\Pi$ is \emph{existentially unforgeable ({UF})}, if for any QPT $\adv$,
    \begin{align*} 
        \Pr [\textsf{Sign}^{\textit{UF}}_{\Pi, \adv}({\lambda})=1]\leq \negl[\lambda].
    \end{align*}
\end{definition}

\subsection{Cryptography from OWF}

In this section, we state relevant results obtained from $\OWF$s. Firstly, Zhandry \cite{Z12} showed how to construct \textsf{pq-PRF} from $\OWF$. Next, \cite{CET21} used \textsf{pq-PRF}s to construct an IND-qCPA-secure symmetric encryption and this was later upgraded to IND-qCCA2 \cite{CEV23}. Informally, these security notions are the same as standard IND-CPA and IND-CCA2, except the adversary in the security experiments is given \emph{quantum} access to the oracles (see \cite{CET21,CEV23} for definitions). 

\begin{lemma}[Symmetric Encryption \cite{CET21,CEV23}]
\label{lem:sym}
    Assuming \textsf{pq-PRF}s, there exists an IND-qCCA2-secure symmetric encryption scheme on classical messages with classical ciphertexts. 
\end{lemma}

We will also need a way to authenticate and encrypt \emph{quantum} messages in a symmetric-key setting. To this end, Alagic, Gagliardoni, and Majenz \cite{AGM18} provided an approach to test \emph{quantum unforgeability} (QUF) and \emph{quantum indistinguishability under adaptive chosen ciphertext attack} (QIND-CCA2) for an authenticated symmetric encryption scheme on quantum messages. Furthermore, they provided a scheme satisfying these notions under \textsf{pq-PRF}s. These experiments are quite technical and involved so we refer the reader to \cite{AGM18} for a detailed description. In Sections \ref{sec:sig def} \&  \ref{sec:bqsm sec}, we present unforgeability experiments adapted to our purposes. 

\begin{lemma}[Symmetric Authenticated Encryption (Theorem 5 in~\cite{AGM18})]
\label{lem:que}
Assuming \textsf{pq-PRF}s, there exists a symmetric key encryption scheme on quantum messages that is QUF and QIND-CCA2. 
\end{lemma}

Also, it is well-known that $\OWF$s imply digital signatures on classical messages. 

\begin{lemma}[Digital Signatures \cite{R90}]\label{lem:ds from owf}
    Digital signatures on classical messages exist assuming $\OWF$s.
\end{lemma}

Note also that an information-theoretically secure \emph{one-time} \textsf{QAE} exists unconditionally.

\begin{lemma}[One-Time Symmetric Authenticated Encryption \cite{BCG02}]
\label{lem:one-time enc}
There exists an information-theoretic unforgeable one-time \textsf{QAE} scheme on quantum messages. In the scheme, if $\lambda$ denotes the security parameter, then a $n$-qubit message is encrypted with a $(n+2\lambda )$-bit key and the resulting ciphertext is a $(n+\lambda)$-qubit state.    
\end{lemma}

\subsection{Cryptography in the BQSM}
\label{sec:def bqs}

In this section, we review major results in the $\BQSM$ based on \cite{BS23}. These results will be useful in constructing the \textsf{BQS} signature scheme. 

We first recall the definition of a \textsf{BQS} one-time program. 

\begin{definition}[\textsf{BQS} One-Time Program \cite{BS23}]
\label{def:BQS one-time}
An algorithm $O$ is a \emph{$(q,\texttt{s})$-\textsf{BQS} one-time compiler} for the class of classical circuits $\mathcal{F}$ if it is QPT and satisfies the following:
\begin{enumerate}
\item (functionality) For any circuit $C\in \mathcal{F}$, the circuit described by $O(C)$ can be used to compute $C$ on a single input $x$ chosen by the evaluator. 
\item For any circuit $C\in \mathcal{F}$, the receiver requires $q$ qmemory to evaluate $O(C)$.
\item (security) For any computationally unbounded adversary $\adv_{\texttt{s}}$ there exists a computationally unbounded simulator $\mathcal{S}_{\texttt{s}}$ such that for any circuit $C\in \mathcal{F}$:
\begin{align*}|\Pr{[\adv_{\texttt{s}}(O(C))=1]}-\Pr{[\mathcal{S}_{\texttt{s}}^{1C}(|0\rangle^{\otimes \lvert C\rvert})=1]}\rvert \leq \negl[\lvert C\rvert] .\end{align*}
\end{enumerate}
\end{definition}

A one-time compiler, denoted as $\ot$, was constructed unconditionally in the $\BQSM$ \cite{BS23}. The details of the construction are not necessary for this work so only the result is given.

\begin{theorem}[\cite{BS23}]
\label{thm:otp}
$\mathcal{O}_{\texttt{s}}$ is a disappearing and unclonable information-theoretically secure $(0,\texttt{s})$ \textsf{BQS} one-time compiler for the class of polynomial classical circuits. 
\end{theorem}

Here, disappearing means the one-time program can no longer be evaluated after the transmission ends. Unclonability means the one-time program cannot be split into two pieces that can both be used to learn an evaluation. The reader is referred to \cite{BS23} for more detailed definitions. 

A similar, more powerful, notion known as \textsf{BQS} program broadcast was also introduced in \cite{BS23}.

\begin{definition}[\textsf{BQS} Program Broadcast]
\label{def BQS program broadcast}
A \emph{$( q,\texttt{s}, k)$-\textsf{BQS} program broadcast} for the class of circuits $\mathcal{C}$ consists of the following QPT algorithms:
\begin{enumerate}
\item $\textsf{KeyGen}(1^\lambda,t_{\textnormal{end}}):$ Outputs a classical key $\ek$. 
\item $\textsf{Br}(1^\s,\ek,C):$ Outputs a quantum transmission $O_C$ for the circuit $C\in \mathcal{C}$ during broadcast time (before $ t_{\textnormal{end}}$). Outputs $\ek$ after broadcast time.
\item $\textsf{Eval}(\langle O_C, \ek\rangle, x):$ Outputs an evaluation $y$ on input $x$ from the transmission $\langle O_C, \ek\rangle$ using $q$ qmemory.
\end{enumerate}
Correctness requires that for any circuit $C\in \mathcal{C}$ and input $x$,
\begin{align*} \Pr{\left[
\begin{tabular}{c|c}
 \multirow{2}{*}{$\textsf{Eval}(\langle O_C, \ek\rangle, x)=C(x)\ $} &  $\ek\ \leftarrow \textsf{KeyGen}(1^\lambda,t_{\textnormal{end}})$ \\ 
 & $O_C\ \leftarrow \textsf{Br}(1^\s,\ek,C)$\\
 \end{tabular}\right]} \geq 1-\negl[\lambda] .
\end{align*}
Security requires that for any (computationally unbounded) adversary $\mathcal{A}_{\texttt{s}}$, there exists a (computationally unbounded) simulator $\mathcal{S}_{\texttt{s}}$ such that for any circuit $C\in \mathcal{C}$, and $\ek\leftarrow \textsf{KeyGen}(1^\lambda,t_{\textnormal{end}})$,
\begin{align*}|Pr[\mathcal{A}_{\texttt{s}}^{\textsf{Br}(1^\s,\ek,C)}(|0\rangle)=1]-Pr[\mathcal{S}_{\texttt{s}}^{kC}(|0\rangle^{\otimes \lvert C\rvert})=1]\rvert \leq \negl[\lambda].\end{align*}
\end{definition}

A program broadcaster, denoted $\Pi_{\textsf{Br}}$, was also constructed unconditionally in the $\BQSM$ \cite{BS23}. Let $\mathcal{C}_{m}$ be the class of polynomial classical circuits with $m$-bit outputs. 

\begin{theorem}[\cite{BS23}]
\label{thm:program broadcast}
$\Pi_{\textsf{Br}}$ is an information-theoretically secure $( 12m, \texttt{s},\frac{\texttt{s}}{2m})$-\textsf{BQS} program broadcast for the class $\mathcal{C}_{m}$ as long as $2^{-m}$ is negligible with respect to the security parameter.
\end{theorem}

\subsection{\textsf{QROM} and \textsf{TLP}}

In the \emph{random oracle model} \cite{BR93}, a hash function $H$ is assumed to return truly random values, and all parties are given access to an oracle of $H$. This model aims to achieve security using the randomness of $H$. In practice, $H$ is replaced with a suitable hash function. However, a quantum adversary can evaluate a hash function on a superposition of inputs. This motivated the introduction of the \emph{quantum random oracle model} ($\QROM$) \cite{BD11}, where users can query the oracle on a superposition of inputs. The reader is referred to \cite{BD11} for more details. 

The important point for our purposes is that the $\QROM$ allows for the creation of time-lock puzzles which were first introduced in \cite{RSW96}. We briefly define this notion based on \cite{BGJ16} and the reader is referred to this work for a detailed definition. 

\begin{definition}[Time-Lock Puzzle \cite{BGJ16}]\label{def:TLP}
   A $\TLP$ with gap $\epsilon<1$ consists of the following pair of algorithms:
    \begin{enumerate}
        \item $\textsf{Puzzle.Gen}(1^\lambda,t, s):$ Generates a classical puzzle $Z$ encrypting the solution {$s\in \{0,1\}^\lambda$} with difficulty parameter $t$.
        \item $\textsf{Puzzle.Sol}(Z):$ Outputs a solution $s\in  \{0,1\}^\lambda$ for the puzzle $Z$. 
    \end{enumerate}
    Correctness requires that for any solution $s$ and puzzle $Z\leftarrow \textsf{Puzzle.Gen}(1^\lambda,t, s)$, $\textsf{Puzzle.Sol}(Z)$ outputs $s$ with probability $1-\negl[\lambda]$. 
    
    Security requires that there exists a polynomial $\tilde{t}$ such that for every polynomial $t\geq \tilde{t}$ and every QPT $\adv$ of depth $\textsf{dep}(\adv)\leq t^\epsilon(\lambda)$, there exists negligible function $\mu$ such that for every $\lambda \in \mathbb{N}^+$ and pair of distinct solutions $s_0,s_1\in \{0,1\}^\lambda$:
    
\begin{align*} \Pr{\left[
\begin{tabular}{c|c}
 \multirow{2}{*}{$\adv(Z_b)=b\ $} &   $b\leftarrow \{0,1\}$ \\ 
  & $Z_b\leftarrow \textsf{Puzzle.Gen}(1^\lambda,t, s_b)$\\
 \end{tabular}\right]} \leq \frac{1}{2}+\mu(\lambda).
\end{align*}
\end{definition}
{
The security condition essentially states that if the puzzle is constructed with difficulty parameter $t$, then the puzzle cannot be deciphered in depth $t^\epsilon$.}

$\TLP$s were constructed classically from non-parallelizing languages and OWFs \cite{BGJ16}. The same construction can be adapted to the post-quantum setting, as noted in \cite{L23}, assuming post-quantum non-parallelizing languages and $\OWF$s. Post-quantum non-parallelizing languages were constructed in the $\QROM$ in \cite{CFH21} giving the following result.

\begin{lemma}[\cite{BGJ16,CFH21}]\label{lem:TLP}
Let $\epsilon <1$. $\TLP$s with gap $\epsilon$ exist assuming $\OWF$s in the $\QROM$. 
\end{lemma}

For the rest of the work, we use $\TLP$s with polynomial $\tilde{t}=0$ in order to simplify notation. Also, we informally say ``a puzzle $Z$ requires $t$ time to solve'' if $Z\leftarrow \textsf{Puzzle.Gen}(1^\lambda,t, s)$ is constructed with difficulty parameter $t$. 

\section{Time-Dependent Signatures}

In this section, we introduce the notion of $\TD$ signatures for quantum messages and present a construction that is secure in the $\QROM$ assuming the existence of $\OWF$s. 

\subsection{Definitions}
\label{sec:TD sig def}

We introduce the notion of $\TD$ signatures with {verification keys that do not change with time}. In this setting, the current time is denoted by $\ttt$ and it is assumed that all parties have access to the current time up to a small error.  

\begin{definition}[Time-Dependent Signatures]
A \emph{$\TD$ signature scheme} $\Pi$ over quantum message space $\hildd{Q}$ consists of the following algorithms: 
\begin{itemize}
    \item $\textsf{KeyGen}(1^\lambda)$: Outputs $(\sk,\vk)$, where $\sk$ is a secret key and $\vk$ is a verification key.
    \item $\textsf{Sign}(\sk,\ttt, \phi):$ Outputs a signature $\sigma$ for $\phi \in \hildd{Q}$ based on the time $\ttt$ using $\sk$. 
    \item $\textsf{Verify}(\vk, \ttt, \sigma')$: Verifies the signature ${\sigma}'$ using $\vk$ and time $\ttt$. Correspondingly outputs a message $\phi'$ or $\perp$. 
\end{itemize}
\end{definition}

\begin{definition}[Correctness]
A $\TD$ signature scheme is \emph{correct} if: for any times $T$ and $T'\leq T+\delta/2$ and $(\sk,\vk)\leftarrow \textsf{KeyGen}(1^\lambda)$: \begin{align*}\| \textsf{Verify}(\vk,T', \textsf{Sign}(\sk,T, \cdot)) - \textsf{id}(\cdot) \|\leq \negl[\lambda] .
 \end{align*}
\end{definition}

We provide a security notion for quantum unforgeability (QUF) that is similar to the security definition for a one-time quantum message authentication introduced in \cite{BCG02}. Essentially, we want the forgery submitted by the adversary to be either invalid or the message extracted from the forgery to be the ``same as'' a state submitted to the signature oracle. By ``same as'' we require that the fidelity between the two states is high (close to 1).

\smallskip \noindent\fbox{%
    \parbox{\textwidth}{%
\textbf{Experiment} $\textsf{TD-Sign}^{\text{QUF}}_{\Pi, \adv}({\lambda})$:
\begin{enumerate}
    \item Experiment $\mathcal{E}$ samples $(\sk,\vk)\leftarrow \textsf{KeyGen}(1^\lambda)$.
    \item Adversary outputs $\tilde{\sigma} \leftarrow \adv^{\ket{\textsf{Sign}(\sk,\ttt, \cdot)}}(\vk)$ at some time $T$. 
    Let $\phi_1,...,\phi_p$ denote the density matrices of the quantum messages $\adv$ submitted to the signature oracle (may be entangled). 
    \item $\mathcal{E}$ computes $ \textsf{Verify}(\vk,T, \tilde{\sigma})$ and obtains either $\perp$ or a state with density matrix $\tilde{\phi}$.
    \item If the outcome is $\perp$, then the experiment is terminated and the output is 0.
    \item Otherwise, let $F\coloneqq \max_{i\in [p]}\delta_F(\phi_i,\tilde{\phi})$ be the maximum fidelity.
    \item $\mathcal{E}$ flips a biased coin that ``produces'' head with probability $F$ and tails with probability $1-F$.
   \item The output of the experiment is $1$ if the coin toss outcome is tails. Otherwise, the output is 0.
\end{enumerate}}}
\smallskip

{Security requires that there is a negligible chance that the coin toss outcome is tails. This translates to the requirement that there is a negligible chance that the state $\tilde{\phi}$ is not very close to one of the submitted states $\phi_1,...,\phi_p$ with respect to the fidelity measure. }

\begin{definition}
    A $\TD$ quantum signature scheme $\Pi$ is \emph{quantum unforgeable ({QUF})}, if for any QPT $\adv$,
    \begin{align*} 
        \Pr [\textsf{TD-Sign}^{\textit{QUF}}_{\Pi, \adv}({\lambda})=1]\leq \negl[\lambda].
    \end{align*}
\end{definition}


Our $\TD$ signatures also satisfy a ``disappearing'' or ``expiration'' property. Advantageously, this property implies that an adversary cannot produce a signature for a message $\phi$ even if it has received a signature of $\phi$ earlier! We present the following experiment to capture this intuition. 

\smallskip \noindent\fbox{%
    \parbox{\textwidth}{%
\textbf{Experiment} $\textsf{TD-Sign}^{\textsf{Dis}}_{\Pi, \adv}({\lambda})$:
\begin{enumerate}
    \item Experiment $\mathcal{E}$ samples $(\sk,\vk)\leftarrow \textsf{KeyGen}(1^\lambda)$.
    \item Adversary gets $\vk$ and quantum oracle access to $\textsf{Sign}(\sk,\ttt, \cdot)$.
    \item $\adv$ outputs 1 at some time $T$ and its query access is revoked. 
    \item $\adv$ outputs $\tilde{\sigma}$.     
    \item $\mathcal{E}$ computes $ \textsf{Verify}(\nk,T+1, \tilde{\sigma})$. 
    \item The output of the experiment is 0 if the result is $\perp$ and 1 otherwise.
\end{enumerate}}}
\smallskip

\begin{definition}
    A $\TD$ quantum signature scheme $\Pi$ is \emph{disappearing} if for any QPT $\adv$,
    \begin{align*} 
         \Pr[\textsf{\TD-Sign}^{\textsf{Dis}}_{\Pi, \adv}({\lambda})=1]\leq \negl[\lambda].
    \end{align*}
\end{definition}

\subsection{Construction}

We present a scheme for $\TD$ signatures that satisfies disappearing and QUF security assuming $\TLP$s and digital signatures on classical messages. Recall, that both these assumptions can be based on $\OWF$s in the $\QROM$ (Lemmas \ref{lem:ds from owf} \& \ref{lem:TLP}). 

\begin{construct}
\label{con:TD sig CV}
{\small Let $\Pi_{\textnormal{1QAE}}\coloneqq (\textsf{1QGen}, \textsf{1QEnc},\textsf{1QDec})$ be the algorithms of a one-time \textsf{QAE} scheme on $n$-qubit messages. Let $\Pi_{\textnormal{CS}}\coloneqq ({\textsf{CS.Gen}}, {\textsf{CS.Sign}},{\textsf{CS.Verify}})$ be the algorithms for a digital signature scheme on classical messages. Let $(\textsf{Puzzle.Gen},\textsf{Puzzle.Sol})$ be the algorithms of a $\TLP$ scheme. The construction for a $\TD$ signature scheme on $n$-qubit messages is as follows:
\begin{itemize}
    \item $\textsf{KeyGen}(1^\lambda):$ Sample $(\sk,\vk) \leftarrow {\textsf{CS.Gen}}(1^\lambda)$. 
 
\item $\textsf{Sign}(\sk,\ttt, \phi):$ 
\begin{enumerate}
    \item Sample $\nk\leftarrow \textsf{1QGen}(1^\lambda)$. 
    \item Compute $\rho \leftarrow \textsf{1QEnc}_{\nk}(\phi)$.
    \item Classically sign $(\ttt, \nk)$ i.e. $\textsf{sig}\leftarrow {\textsf{CS.Sign}}(\sk, (\ttt,\nk))$. 
    \item Generate $\TLP$ puzzle $Z\leftarrow \textsf{Puzzle.Gen}(1^\lambda, 1,(\ttt, \nk, \textsf{sig}))$.
\end{enumerate}
Output $\sigma \coloneqq (\rho, Z)$.

\item $\textsf{Verify}(\vk,\ttt, \sigma'):$ 
\begin{enumerate}
    \item Interpret $\sigma'$ as $(\rho',Z')$ and take note of the current time as $T$. 
    \item Compute $(T', \nk', \textsf{sig}')\leftarrow \textsf{Puzzle.Sol}(Z').$
    \item If $T' +\delta/2\leq T$ or $\textsf{CS.Verify}(\vk,(T', \nk'), \textsf{sig}')=0 $, then output $\perp$.
    \item Otherwise, output $\textsf{1QDec}_{\nk'}(\rho')$.
\end{enumerate}
\end{itemize}}
\end{construct}

\begin{theorem}
    Construction \ref{con:TD sig CV} is a QUF and disappearing $\TD$ signature scheme assuming the existence of secure $\TLP$s and EUF digital signatures on classical messages. 
\end{theorem}

\begin{proof}
We first prove QUF security through a sequence of hybrids. We let hybrid $\textsf{H}_0$ be the standard experiment, but adapted to our construction for an easier discussion. 

\smallskip \noindent\fbox{%
    \parbox{\textwidth}{%
\textbf{Hybrid} $\Hy_0$ (standard experiment):
\begin{enumerate}
    \item $\mathcal{E}$ samples $(\sk,\vk)\leftarrow \textsf{KeyGen}(1^\lambda)$.
    \item Adversary outputs $\tilde{\sigma}\leftarrow \adv^{\ket{\textsf{Sign}(\sk,\ttt, \cdot)}}(\vk)$ at time $T$. 
    \item Let $\phi_1,...,\phi_p$ denote the density matrices of the quantum messages $\adv$ submitted to the signature oracle and let $(\rho_i,Z_i)_{i\in [p]}$ denote the responses.
    \item $\mathcal{E}$ verifies the forgery as follows:
\begin{enumerate}
    \item Interpret $\tilde{\sigma}$ as $(\tilde{\rho},\tilde{Z})$. 
    \item Compute $(\tilde{T}, \tilde{\nk}, \tilde{\textsf{sig}})\leftarrow \textsf{Puzzle.Sol}(\tilde{Z}).$
    \item If $\tilde{T} +\delta/2\leq T$ or $\textsf{CS.Verify}(\vk,(\tilde{T}, \tilde{\nk}), \tilde{\textsf{sig}})=\perp$, then output $\perp$.
    \item Otherwise, compute $\tilde{\phi}\leftarrow \textsf{1QDec}_{\tilde{\nk}}(\tilde{\rho})$.
\end{enumerate}
     \item Let $F\coloneqq \max_{i\in [p]}\delta_F(\phi_i,\tilde{\phi})$ be the maximum fidelity.
    \item $\mathcal{E}$ flips a biased coin that ``produces'' head with probability $F$ and tails with probability $1-F$.
   \item The output of the experiment is $1$ if the coin toss outcome is tails. Otherwise, the output is 0.
\end{enumerate}}}
\smallskip

In the next hybrid, the experiment rejects any forgery which includes a signature of a new message under the classical signature scheme. 

\smallskip \noindent\fbox{%
    \parbox{\textwidth}{%
\textbf{Hybrid} $\Hy_1$:
\begin{enumerate}
    \item Steps 1-3 same as $\Hy_0$.
\setcounter{enumi}{3}
\item $\mathcal{E}$ verifies the forgery as follows:
\begin{enumerate}
    \item Interpret $\tilde{\sigma}$ as $(\tilde{\rho},\tilde{Z})$. 
    \item Compute $(\tilde{T}, \tilde{\nk}, \tilde{\textsf{sig}})\leftarrow \textsf{Puzzle.Sol}(\tilde{Z}).$
    \item Compute $(T_i, \nk_i, \textsf{sig}_i)\leftarrow \textsf{Puzzle.Sol}(Z_i)$ for all $i\in [p]$.
    \item If $\tilde{T} +\delta/2\leq T$ or $(\tilde{T}, \tilde{\nk})\notin (T_i,\nk_i)_{i\in [p]}$, then output $\perp$.
    \item Otherwise, compute $\tilde{\phi}\leftarrow \textsf{1QDec}_{\tilde{\nk}}(\tilde{\rho})$.
\end{enumerate}
     \item Steps 5-7 same as $\Hy_0$.
\end{enumerate}}}
\smallskip

Hybrid $\Hy_2$ is the same as hybrid $\Hy_1$ except we discard the classical signatures. The new approach to sign is as follows:

$\textsf{Sign}_{\Hy_2}(\sk,\ttt, \phi)$:
\begin{enumerate}
    \item Sample $\nk\leftarrow \textsf{1QGen}(1^\lambda)$. 
    \item Compute $\rho \leftarrow \textsf{1QEnc}_{\nk}(\phi)$.
    \item Generate $\TLP$ puzzle $Z\leftarrow \textsf{Puzzle.Gen}(1^\lambda, 1,(\ttt, \nk))$.
    \item Output $\sigma \coloneqq (\rho, Z)$.
\end{enumerate}

Next, hybrid $\textsf{H}_3$ is the same as $\textsf{H}_2$, except we change the signing and verification algorithms again but only for messages close to time $T$. Specifically, the algorithms are as follows:

\begin{itemize}
    \item $\textsf{KeyGen}_{\Hy_3}(1^\lambda):$ Sample $(\sk,\vk) \leftarrow {\textsf{CS.Gen}}(1^\lambda)$. 
\item $\textsf{Sign}_{\Hy_3}(\sk,\ttt, \phi_i)$: If $\ttt< T-\delta/2$, then output $\textsf{Sign}_{\Hy_2}(\sk,\ttt, \phi_i)$. Otherwise:
\begin{enumerate}
    \item Sample $\nk_i\leftarrow \textsf{1QGen}(1^\lambda)$. 
    \item Store $\nk_i$. 
    \item Compute $\rho_i \leftarrow \textsf{1QEnc}_{\nk_i}(\phi_i)$.
    \item Generate $\TLP$ puzzle $Z_i\leftarrow \textsf{Puzzle.Gen}(1^\lambda, 1,(i,0^{\lvert \nk_i\rvert}))$.
    \item Output $(\rho_i, Z_i)$.
\end{enumerate}
\item $\textsf{Verify}_{\Hy_3}(\vk,\ttt, \sigma'):$ 
\begin{enumerate}
\item If $\ttt <T$, output $\textsf{Verify}_{\Hy_2}(\vk,\ttt, \sigma').$ 
    \item Otherwise, interpret $\tilde{\sigma}$ as $(\tilde{\rho},\tilde{Z})$. 
    \item Compute $\textsf{Puzzle.Sol}(\tilde{Z}).$
    \item If the solution is not of the form $(j,0^*)$ for some $j\in [p]$, then output $\perp.$
    \item Otherwise, obtain $\nk_j$ from storage and output $\tilde{\phi}\leftarrow \textsf{1QDec}_{{\nk}_j}(\tilde{\rho})$.
\end{enumerate}
\end{itemize}

\begin{claim}
    $\Pr[\textsf{H}_3=1]\leq \negl[\lambda]$.
\end{claim}

\begin{proof}
In this hybrid, the experiment verifies the forgery $\tilde{\sigma}$ as follows. It first interprets $\tilde{\sigma}$ as $(\tilde{\rho},\tilde{Z})$. Then, it computes $\textsf{Puzzle.Sol}(\tilde{Z})$ and if the solution is not of the form $(j,0^*)$ for some $j\in [p]$, then it outputs $\perp.$ In the case this condition is satisfied, then the experiment runs $ \textsf{1QDec}_{{\nk}_j}(\tilde{\rho})$. Note that the only information available regarding $\nk_j$ is the authenticated state $\rho_j$ from the $j^{\textsf{th}}$ query to the signature oracle. Hence, if an adversary manages to produce a fresh forgery, then $\adv$ can easily be converted into an a QPT adversary that breaks QUF of $\Pi_{1QAE}$.
\qed
\end{proof}

\begin{claim}
    No QPT adversary can distinguish hybrids $\Hy_3$ and $\Hy_2$ with non-negligible advantage.
\end{claim}

\begin{proof}
The only difference between these hybrids is the information contained in the $\TLP$s produced after time $T-\delta/2$. The $\TLP$s require an hour to solve and $\adv$ submits its output at time $T$. Hence, by the security of $\TLP$s, no adversary can distinguish between these hybrids.
\qed
\end{proof}

\begin{claim}
    The output of $\Hy_2$ and $\Hy_1$ are indistinguishable.
\end{claim}
\begin{proof}
There is no difference between the output of these hybrids as the classical signature scheme used in $\Hy_1$ is redundant.
\qed
\end{proof}

\begin{claim}
    No QPT adversary can distinguish hybrids $\Hy_1$ and $\Hy_0$ with non-negligible advantage.
\end{claim}
\begin{proof}
The only difference between these hybrids is that in $\Hy_0$, the output of the experiment is 0 if $\textsf{CS.Verify}(\vk,(\tilde{T}, \tilde{\nk}), \tilde{\textsf{sig}})=\perp$, whereas in $\Hy_1$, this condition is replaced with $(\tilde{T}, \tilde{\nk})\notin (T_i,\nk_i)_{i\in [p]}$. Hence, if $\adv$ can distinguish between these hybrids, then with non-negligible probability $\adv$'s forgery satisfies $(\tilde{T}, \tilde{\nk})\notin (T_i,\nk_i)_{i\in [p]}$ and $\textsf{CS.Verify}(\vk,(\tilde{T}, \tilde{\nk}), \tilde{\textsf{sig}})=1$. Such an algorithm can easily be converted into a QPT algorithm that breaks the unforgeability of $\Pi_{\textnormal{CS}}$.
\qed
\end{proof}

The past three claims imply that $\Pr[\Hy_0=1]\leq \negl[\lambda]$. 

Disappearing security is similar but easier to show. In the disappearing experiment $\textsf{TD-Sign}^{\textsf{Dis}}_{\Pi, \adv}({\lambda})$, $\adv$ submits the forgery after a time delay of 1 from the submission of any signature query i.e. $T\geq T_i+1$ for any $i\in [p]$. However, we showed that to pass verification, (whp) there exists $j\in [p]$ such that $\tilde{T}=T_j$ which means $T\geq \tilde{T}+1$. So (whp) the forgery does not pass verification and the output of the experiment is 0. 
\qed
\end{proof}

\section{Time-Dependent Signatures with Dynamic Verification Keys}

In this section, we first define $\TdVK$ signatures which are $\TD$ signatures but where the verification key varies with time. Then, we provide a construction for this primitive assuming $\OWF$s. 

\subsection{Definitions}
\label{sec:sig def}

\begin{definition}[Time-Dependent Signatures with Dynamic Verification Keys]
A \emph{$\TdVK$ signature scheme} over quantum message space $\hildd{Q}$ starting at time $t_0$ with interval $\delta$ consists of the following algorithms: 
\begin{itemize}
    \item $\textsf{KeyGen}(1^\lambda)$: Outputs a time-independent secret key $\sk$.
    \item $\textsf{vkAnnounce}(\sk,\ttt):$ {If time $ \ttt =(i+1)\delta +t_0$ for some $i\in \mathbb{N}$}, outputs a verification key $\vk_i$ using the secret key $\sk$. 
    \item $\textsf{Sign}(\sk,\ttt, \phi):$ Outputs a signature $\sigma$ for $\phi \in \hildd{Q}$ based on the time $\ttt$ using $\sk$. 
    \item $\textsf{Verify}(\vk_i, \sigma')$: Verifies the signature ${\sigma}'$ using $\vk_i$ and correspondingly outputs a message $\phi'$ or $\perp$. 
\end{itemize}
\end{definition}

Henceforth, we let $t_i\coloneqq i\delta+t_0$ and we generally set $\delta=1$ and $t_0=0$ without loss of generality. 

\begin{definition}[Correctness]
A $\TdVK$ signature scheme is \emph{correct} if: for any $i\in\mathbb{N}$, time $\ttt\in [t_i,t_{i+1})$, $\sk\leftarrow \textsf{KeyGen}(1^\lambda)$ and $\vk_i\leftarrow \textsf{vkAnnounce}(\sk,t_{i+1})$ we have:
 \begin{align*}\| \textsf{Verify}(\vk_i,\textsf{Sign}(\sk,\ttt,\cdot)) - \textsf{id}(\cdot) \|\leq \negl[\lambda] .
 \end{align*}
\end{definition}

Security requires that any forgery submitted is close to a signature produced by the signing oracle in the same time frame as the forgery. Note that we assume that the verification keys $\vk_i$ can be distributed authentically. This can be ensured through the use of any digital signature scheme on classical messages under a time-independent verification key. 

\smallskip \noindent\fbox{%
    \parbox{\textwidth}{%
\textbf{Experiment} $\textsf{TDDVK-Sign}^{\text{QUF}}_{\Pi, \adv}({\lambda},i)$:
\begin{enumerate}
   \item Experiment $\mathcal{E}$ samples $\sk\leftarrow \textsf{KeyGen}(1^\lambda)$. 
   \item For all $j< i$, $\mathcal{E}$ generates $\vk_j\leftarrow \textsf{vkAnnounce}(\sk,t_{j+1})$.     
   \item Adversary outputs $\tilde{\sigma} \leftarrow \adv^{\ket{\textsf{Sign}(\sk,\cdot, \cdot)}}((\vk_j)_{j< i})$. 
    Let $(T_1,\phi_1),...,(T_p,\phi_p)$ denote the subset of messages submitted to the signing oracle which satisfy $t_i\leq T_k<t_{i+1}$ for $k\in [p]$. 
    \item $\mathcal{E}$ generates $\vk_i \leftarrow \textsf{vkAnnounce}(\sk,t_{i+1})$.
    \item $\mathcal{E}$ computes $\tilde{\phi}\leftarrow  \textsf{Verify}(\vk_i, \tilde{\sigma})$.
    \item If the outcome is $\perp$, then the experiment is terminated and the outcome is 0.
    \item Otherwise, let $F\coloneqq \max_{i\in [p]}\delta_F(\phi_i,\tilde{\phi})$ be the maximum fidelity.
    \item $\mathcal{E}$ flips a biased coin that ``produces'' head with probability $F$ and tails with probability $1-F$.
   \item The output of the experiment is $1$ if the coin toss outcome is tails. Otherwise, the output is 0.
\end{enumerate}}}
\smallskip

\begin{definition}
    A $\TdVK$ quantum signature scheme $\Pi$ is \emph{quantum unforgeable ({QUF})}, if for any QPT $\adv$ and $i\in \mathbb{N}$,
    \begin{align*} 
        \Pr [\textsf{TDDVK-Sign}^{\textit{QUF}}_{\Pi, \adv}({\lambda},i)=1]\leq \negl[\lambda].
    \end{align*}
\end{definition}

\subsection{One-Time \textsf{TD-DVK} Signature Scheme}
\label{sec:one-time sig}

We first present a very simple information-theoretic $\TdVK$ one-time signature scheme. The basic idea is to let the verification key be the same as the signing key but with a \emph{delay}. In the next section, we extend it to allow for many signatures using $\OWF$s.

{In this scheme, a single message can be signed until a certain time, which we denote $t_{\textnormal{end}}$, chosen by the signer. After this time, the verification key is announced and no new signatures can be securely generated. }While such a construction seems trivial, it has practical and useful applications as shown in Sec.~\ref{sec:pure-state}. In particular, it allows for the authenticated distribution of pure-state quantum keys in $\QPKE$ schemes answering one of the open questions posed in an earlier work \cite{KMN23}. 

\begin{construct}
\label{con:one-time sig}
{\small Let $\Pi_{\textnormal{QAE}}\coloneqq (\textsf{1QGen}, \textsf{1QEnc},\textsf{1QDec})$ be the algorithms for a symmetric one-time \textsf{QAE} scheme on $n$-qubit messages. The construction for a one-time $\TdVK$ signature scheme on $n$-qubit messages with end time $t_{\textnormal{end}}$ is as follows:
\begin{itemize}
    \item $\textsf{KeyGen}(1^\lambda):$ Sample $\nk \leftarrow \textsf{1QGen}(1^\lambda)$. 
     \item $\textsf{vkAnnounce}(\nk, \ttt):$ If $\ttt=t_{\textnormal{end}}$, then announce $\vk\coloneqq (t_{\textnormal{end}},\nk)$. 
\item $\textsf{Sign}(\nk,\phi):$ Output state $\sigma \leftarrow \textsf{1QEnc}_{\nk}(\phi).$
\item $\textsf{Verify}(\vk, \sigma'):$ Interpret $\vk$ as $(t_{\textnormal{end}},\nk)$. If $\sigma'$ is received after time $t_{\textnormal{end}}$, then output $\perp$. Otherwise, output $\textsf{1QDec}_{\nk}(\sigma').$
\end{itemize}}
\end{construct}

\subsection{Many-Time \textsf{TD-DVK} Signature Scheme}

In this section, we build a (many-time) $\TdVK$ signature scheme from $\OWF$s. In a time frame $[t_i,t_{i+1})$, messages are authenticated with a symmetric \textsf{QAE} scheme using a key $\nk_i$ which is generated from the secret key and a pseudorandom function. Then, $\nk_i$ is revealed at time $t_{i+1}$ allowing receivers to verify transmissions received in the specified time frame. Following this, new messages are signed with $\nk_{i+1}$, and the process is repeated.

Recall, the existence of $\OWF$s implies the existence of \textsf{pq-PRF}s and a QUF \textsf{QAE} scheme (Lemma \ref{lem:que}).

\begin{construct}
\label{con:Sig}
{\small Let $\Pi_{\textnormal{QAE}}\coloneqq (\textsf{QGen},\textsf{QEnc},\textsf{QDec})$ be the algorithms of a QUF \textsf{QAE} scheme on $n$-qubit messages where the key generation algorithm makes $p$ coin tosses. Let $F:\{0,1\}^\lambda \times \{0,1\}^n\rightarrow \{0,1\}^p$ be a \textsf{pq-PRF}. The construction for $\TdVK$ signatures on $n$-qubit messages is as follows:
\begin{itemize}
    \item $\textsf{KeyGen}(1^\lambda):$ Sample $\sk \leftarrow \{0,1\}^\lambda$. 
     \item $\textsf{vkAnnounce}(\sk, \ttt):$ If $\ttt=t_{i+1}$ for some $i\in \mathbb{N}$, then compute $y_i\coloneqq F(\sk,i)$ and announce $\vk_i\coloneqq (t_i,t_{i+1},y_i)$. 
\item $\textsf{Sign}(\sk,\ttt, \phi):$ Let $j\in \mathbb{N}$ be such that $\ttt\in [t_{j},t_{j+1})$. 
\begin{enumerate}
    \item Compute $y_j\coloneqq F(\sk,j)$.
    \item Generate $\nk_j\leftarrow \textsf{QGen}(1^\lambda;y_j)$.
    \item Output $\sigma \leftarrow \textsf{QEnc}_{\nk_j}(\phi).$
\end{enumerate}

\item $\textsf{Verify}(\vk_i, \sigma'):$ 
\begin{enumerate}
    \item Interpret $\vk_i$ as $(t_i,t_{i+1},y_i)$. 
    \item Generate $\nk_i\leftarrow \textsf{QGen}(1^\lambda;y_i)$. 
    \item Output $\textsf{QDec}_{\nk_i}(\sigma').$
\end{enumerate} 
\end{itemize}}
\end{construct}

\begin{theorem}
    Construction \ref{con:Sig} is a quantum unforgeable $\TdVK$ signature scheme assuming the existence of a QUF \textsf{QAE} and \textsf{pq-PRF}s. 
\end{theorem}

\begin{proof}
    Fix $i\in \mathbb{N}$ and consider the unforgeability experiment $\textsf{TDVK-Sign}^{\text{QUF-Forge}}_{\Pi, \adv}({\lambda},i)$. $\adv$ is given oracle access to $\textsf{Sign}(\sk,\cdot,\cdot)$ and tries to output a signature that is verified with key $\vk_i$, specifically by using $\nk_i=F(\sk,i)$. 
    
    By the security of \textsf{pq-PRF}s, we can switch to a hybrid where the keys $(\nk_j)_{j\leq i }$ are chosen uniformly at random. In particular, $\adv$ cannot distinguish $\nk_i$ from random and must forge a signature that is verified with key $\nk_i$. $\adv$ has oracle access to $\textsf{QEnc}_{\nk_i}(\cdot)$ and $\textsf{QDec}_{\nk_i}(\cdot)$ but this is the only information $\adv$ can receive related to $\nk_i$. If the output of the experiment is 1 with non-negligible probability, then $\adv$ can easily be converted into an adversary that brakes the QUF-security of $\Pi_{\textnormal{QAE}}$. 
    \qed 
\end{proof}

It seems that Construction \ref{con:Sig} is insecure after time $t_{2^n}$ since the keys start repeating. However, this issue can easily be resolved using the pseudorandom function. {For instance, for any integer $p$, by letting $\nk_{i+2^{pn}}\coloneqq F^p(\sk,i)$, where $F^p$ denotes $p$ applications of $F$, the keys are unpredictable at any time. In other words, we can keep producing seemingly random keys indefinitely. }

\section{Signatures in the Bounded Quantum Storage Model}
\label{sec:bqsm sig}
In this section, we present a signature scheme for quantum messages that is information-theoretically secure against adversaries with bounded quantum memory. We note that signing and verification algorithms in this setting do not depend on time. The public-keys are mixed-states and are certified with classical certification keys.

\subsection{Definitions}
\label{sec:bqsm sec}

We define the notion of \textsf{BQS} signatures on quantum messages. Such constructions are built against adversaries $\adv_{\s}$ with qmemory bound $\s$. The public keys are distributed for a limited time only. In fact, this is unavoidable for information-theoretic security as noted in \cite{BS23}. We note that the signature scheme on classical messages in the bounded \emph{classical} storage model suffers from this limitation as well \cite{DQW22}. 

\begin{definition}[\textsf{BQS} Signature Scheme]
A \emph{\textsf{BQS} signature scheme} $\Pi$ over quantum message space $\hildd{Q}$ against adversaries with memory bound $\s$, with key distribution end time $t_{\textnormal{end}}$, consists of the following algorithms: 
\begin{itemize}
    \item $\textsf{KeyGen}(1^\lambda,1^\s)$: Outputs a secret key $\sk$.
    \item $\textsf{ckAnnounce}(\sk,\ttt)$: Outputs a classical certification key $\ck$ at time $t_{\textnormal{end}}$.
    \item $\textsf{pkSend}(1^\s,\sk):$ Outputs a quantum public-key $\rho$.
    \item $\textsf{vkReceive}(\rho,\ck):$ Processes $\rho$ using $\ck$ and outputs a classical verification key $\vk$.
    \item $\textsf{Sign}(\sk,\phi):$ Outputs a signature $\sigma$ for $\phi \in \hildd{Q}$ using $\sk$. 
    \item $\textsf{Verify}(\vk, \sigma')$: Outputs a message $\phi'$ or $\perp$ using $\vk$. 
\end{itemize}
\end{definition}

\begin{definition}[Correctness]
A \textsf{BQS} signature scheme $\Pi$ is \emph{correct} if for any $\sk\leftarrow \textsf{KeyGen}(1^\lambda,1^\s)$, $\rho \leftarrow \textsf{pkSend}(1^\s,\sk)$, $\ck\leftarrow \textsf{ckAnnounce}(\sk,t_{\textnormal{end}})$, and $\vk\leftarrow \textsf{vkReceive}(\rho,\ck)$:
 \begin{align*}\| \textsf{Verify}(\vk,\textsf{Sign}(\sk,\cdot)) - \textsf{id}(\cdot) \|\leq \negl[\lambda] .
 \end{align*}
\end{definition}

We introduce a security notion for \textsf{BQS} signature schemes which allows the adversary unrestricted access to the signing oracle and the verification oracle under the verification key that the experiment extracts. We also give the adversary the ability to tamper with the quantum public-keys, meaning the adversary provides the quantum key which the experiment uses to extract its verification key. We call our notion \emph{quantum unforgeability under chosen key attack (BQS-QUCK)}.

\smallskip \noindent\fbox{%
    \parbox{\textwidth}{%
\textbf{Experiment} $\textsf{BQS-Sign}^{\text{QUCK}}_{\Pi, \adv}({\lambda},\s)$:
\begin{enumerate}
    \item Experiment $\mathcal{E}$ samples $\sk\leftarrow \textsf{KeyGen}(1^\lambda,1^\s)$.
    \item $\tilde{\rho}\leftarrow \adv^{\textsf{pkSend}(1^\s,\sk)}$.
    \item $\mathcal{E}$ announces certification key $\ck\leftarrow \textsf{ckAnnounce}(\sk,t_{\textnormal{end}})$.
    \item $\mathcal{E}$ extracts key $\nk\leftarrow \textsf{vkReceive}(\tilde{\rho},\ck)$.
    \item $\tilde{\sigma} \leftarrow \adv^{\ket{\textsf{Sign}(\sk,\cdot)},\ket{\textsf{Verify}(\nk,\cdot)}}$. 
    Let $\phi_1,...,\phi_p$ denote the density matrices of the quantum messages $\adv$ submitted to the signature oracle. 
    \item $\mathcal{E}$ computes $\tilde{\phi}\leftarrow \textsf{Verify}(\nk,\tilde{\sigma})$.
    \item If $\tilde{\phi}=\perp$, then the experiment is terminated and the outcome is 0.
    \item Otherwise, let $F\coloneqq \max_{i\in [p]}\delta_F(\phi_i,\tilde{\phi})$ be the maximum fidelity.
    \item $\mathcal{E}$ flips a biased coin that ``produces'' head with probability $F$ and tails with probability $1-F$.
   \item The output of the experiment is $1$ if the coin toss outcome is tails. Otherwise, the output is 0.
\end{enumerate}}}
\smallskip

\begin{definition}
A \textsf{BQS} signature scheme $\Pi$ satisfies \emph{quantum unforgeability under chosen key attack (BQS-QUCK)} if for $\s\in \mathbb{N}$ and any (computationally-unbounded) adversary $\advs$,
    \begin{align*}
         \Pr[\textsf{BQS-Sign}^{\textit{QUCK}}_{\Pi, \advs}({\lambda},\s)=1] \leq \negl[\lambda].
    \end{align*}
\end{definition}

\subsection{Construction}

In this section, we show how to sign quantum messages in the $\BQSM$. The quantum public-keys consist of a program broadcast of an appropriate function. A subset of the function image is revealed after the broadcast to allow users to certify their evaluations. The signature of a message consists of two one-time programs followed by an \textsf{QAE} of the message state. The one-time programs reveal the necessary information to verify the state only when evaluated on inputs corresponding to correct evaluations of the broadcasted program. The scheme requires the honest user to have $O(l^2)$ qmemory where $l$ is the size of the quantum messages to be signed. 

\begin{construct}
\label{con:cca}
{\small Let $\Pi_{\textnormal{1QAE}}=(\textsf{1QGen},\textsf{1QEnc}, \textsf{1QDec})$ be the algorithms of a one-time \textsf{QAE} scheme over $l$-qubit messages where the key generation makes $\ell$ coin tosses. Let $\Pi_{\textsf{Br}}=(\textsf{KeyGen},\textsf{Br},\textsf{Eval})$ be the algorithms for a program broadcast. Let $t_{\textnormal{end}}$ denote the time when key distribution phase ends. The construction for \textsf{BQS} signature scheme over $l$-qubit messages is as follows:

\begin{itemize}
   \item $\textsf{KeyGen}(1^\lambda,1^\texttt{s})$: 
   \begin{enumerate}
       \item Let $m=100\ell^2$ and $n= \max(\lceil \frac{\texttt{s}}{2m}\rceil+1,m)$. 
       \item Choose uniformly at random a $  m  \times n$ binary matrix $M$ and define $P(x)\coloneqq M\cdot x$. 
       \item Sample an evaluation key $\ek\leftarrow \textsf{KeyGen}(1^\lambda,t_{\textnormal{end}})$.
       \item Output $\sk=(P,\ek)$
   \end{enumerate} 

    \item $\textsf{pkSend}(1^\s, \sk)$:
    Broadcast the program $P$, i.e. output $\rho \leftarrow \textsf{Br}(1^\s,\ek,P)$. 

    \item $\textsf{ckAnnounce}( \sk,\ttt)$: 
    \begin{enumerate}
        \item If $\ttt\neq t_{\textnormal{end}}$, then output $\perp$.
        \item Otherwise, choose a random subset $S$ of $[m]$ of size $\frac{m}{2} $ and hash function $H\leftarrow \mathbf{H}_{m/2,m/8}$.
        \item Let $P_S$ be the matrix consisting of rows in $P$ whose index is in $S$. 
        \item Output $\ck\coloneqq (t_{\textnormal{end}}, S, P_S,H)$.
    \end{enumerate}  

    \item $\textsf{vkReceive}( \ck,\rho,\ttt)$:
    Process the public and certification keys as follows:
    \begin{enumerate}
        \item Interpret $\ck$ as $(t_{\textnormal{end}}, S, P_S,H)$. 
        \item If $\rho$ is received after $t_{\textnormal{end}}$, then output $\perp$.
        \item Sample $v\leftarrow \{0,1\}^n$. 
        \item Compute $\textsf{Eval}(\rho, v)$ and let $P'(v)$ denote the output.
        \item If the entries in $P'(v)$ whose index is in $S$ do not agree with $P_S(v)$, then output $\perp$.
        \item Otherwise, let $\overline{S}$ be the complement of $S$ and let $P_v\coloneqq H(P_{\overline{S}}'(v))$. 
        \item Output $\nk_v\coloneqq (v,P_v)$.
    \end{enumerate}
    
    \item $\textsf{Sign}( \sk, \phi)$: 
    \begin{enumerate}
        \item Sample $\nk\leftarrow \textsf{1QGen}(1^\lambda)$ and $a,b,a',b',a'',b''\leftarrow \{0,1\}^{\ell}$.
        \item Define $f(x) \coloneqq ax+b \mod{2^{\ell}}$ and similarly define $f'$ and $f''$. 
        \item Let $F(x)\coloneqq (f(x),f'(x),f''(x))$ and let $P_x\coloneqq H(P_{\overline{S}}(x))$.
        \item Construct a program $S_{P,F,\nk}$ which takes as input $(x,y)\in \{0,1\}^{n,2\ell} $ and does as follows:
        \begin{enumerate}
            \item Compute $P_x$ and interpret the resulting string as an encoding for $M_x\in \{0,1\}^{2\ell \times 2\ell}$, $M_x'\in \{0,1\}^{3\ell \times 3\ell}$, $z_x\in \{0,1\}^{2\ell}$ and $z_x'\in \{0,1\}^{3\ell}$.
            \item If $y=M_x\cdot (t\|f(t))+z_x$ for some $t\in \{0,1\}^{\ell}$, then output $(M_x'\cdot (f(t)\|f'(t)\|f''(t))+z'_x,f'(t)\nk+f''(t), \nk)$.
            \item Output $\perp$ otherwise.
        \end{enumerate}
        \item Output the signature state $\sigma \leftarrow  \langle \mathcal{O}_{\texttt{s}}(F),\mathcal{O}_{\texttt{s}}(S_{P,F,\nk}),\textsf{1QEnc}_{\nk}(\phi )\rangle$ sent in sequence. 
    \end{enumerate}

\item $\textsf{Verify}(\nk_v, \sigma)$: 
\begin{enumerate}
    \item Interpret $\nk_v$ as $(v,Y_v)$.
    \item Check the signature is of the correct form and output $\perp $ if it is not. 
    \item Evaluate the first one-time program on a random input $t\leftarrow \{0,1\}^\ell$ and let $(f(t),f'(t),f''(t))$ denote the result. 
    \item Use $Y_v$ to deduce $(M_v,M_v',z_v,z_v')$ and evaluate the second program on $(v, M_v \cdot (t\|f(t))+z_v)$.
    \item If the output is $(M_v'\cdot (f(t)\|f'(t)\|f''(t))+z'_v ,f'(t)\nk'+f''(t),\nk')$ for some $\nk'\in \{0,1\}^\ell$, then output $\textsf{1QDec}_{\nk'}(|\overline{c}\rangle)$ where $|\overline{c}\rangle$ is the final state in $\sigma$ and otherwise output $\perp$. 
\end{enumerate}
  
\end{itemize}}\end{construct}

\remark{In the $\textsf{Verify}$ algorithm, the signature is checked to be of the correct form. More explicitly, this means the signature is checked to be of the form $\langle \mathcal{O}_{\texttt{s}}(F),\mathcal{O}_{\texttt{s}}(S),\overline{c}\rangle$ for $\mathbf{NC}^1$ functions ${F}:\{0,1\}^\ell \rightarrow \{0,1\}^{\ell,\ell,\ell}$ and ${S}: \{0,1\}^{n,2\ell} \rightarrow  \{0,1\}^{3\ell,\ell,\ell}$. If a query submitted to the verification oracle is not of this form, then it is rejected. One-time programs satisfy this form of receiver security as shown in \cite{BS23}.}

{
{We briefly give an intuition behind the security of the scheme before presenting the formal proof. Intuitively, the signing algorithm consists of two one-time programs and an encryption of the message. The first one-time program is independent of the secret key and consists of random linear functions, and the second only returns non-$\bot$ outputs on inputs that involve the secret key. Therefore, it seems that the  adversary, not knowing the secret key, cannot learn anything from these one-time programs. The problem is that the adversary  may take a valid signature and edit in some way so that it signs another message. Intuitively, this attack is thwarted using this dual one-time program structure.  Notice that $F$ consists of  linear functions so the evaluation on any input is indistinguishable from random when given a single evaluation of another input. In the scheme, $F$ is sampled fresh at random for every signature and only a single one-time program of $F$ is provided. In other words, an adversary can learn a single evaluation and all other evaluations of $F$ are indistinguishable from random. In the second one-time program, for valid inputs, we essentially perform an authentication of a randomly sampled key $\nk$ along with an output from the first one-time program using a part of the secret key. The key $\nk$ is used to perform an authenticated encryption of the quantum message. It is difficult for the adversary to meaningfully edit the first one-time program (without failing verification) since its output is authenticated in the second one-time program. At the same time, it is difficult to edit the output of the second one-time program since the adversary only has knowledge of a single evaluation of $F$. In essence, the two one-time programs are interdependent in a way that makes it difficult to meaningfully modify either of them. }
}

\begin{theorem}
\label{construction 1}
Construction \ref{con:cca} is a quantum unforgeable \textsf{BQS} signature scheme.
\end{theorem}

\begin{proof}
In the $\textsf{BQS-Sign}^{\text{QUCK}}_{\Pi, \adv_\s}({\lambda})$ experiment, the adversary $\advs$ outputs a key $\tilde{\rho}\leftarrow \advs^{\textsf{pkSend}(1^\s,\sk)}(\ck)$ and, then the experiment extracts a key $\nk_v \leftarrow \textsf{vkReceive}(\tilde{\rho},\ck)$ for some random $v$. $\advs $ must send $\tilde{\rho}$ before $t_{\textnormal{end}}$ or else $\nk_v=\perp$. In the case $\tilde{\rho}$ is received before $t_{\textnormal{end}}$, $\mathcal{E}$ obtains an evaluation from the public-key copy which we denote by $\tilde{P}(v)$. Note that with high probability either $\nk_v=\perp$ or $\|\tilde{P}(v)-P(v)\|\leq n/8$. This is because if $\tilde{P}(v)$ differs from $P(v)$ on more than $n/8$ positions, then there is less than $2^{-n/8}$ chance that none of these positions are a subset of $S$. If one of the positions is in $S$, then $\nk_v=\perp$ by definition of $\textsf{vkReceive}$. Hence, we can assume that $\tilde{\rho}$ is received before $t_{\textnormal{end}}$ and that $\|\tilde{P}(v)-P(v)\|\leq n/8$.

By the security of program broadcast (Theorem \ref{thm:program broadcast}), $\advs$'s access to the public-key copies can be simulated with $\frac{\texttt{s}}{2m}$ queries to an oracle of ${P}$. By Lemma \ref{Raz 2}, since $M$ is a matrix of dimension $ m \times n$ and $n> \frac{\texttt{s}}{2m}$, $\advs$ cannot guess the output of $P(x)=M\cdot x$ on a random input except with negligible probability. In other words, if $K_v$ represents the random variable for $P(v)$ from the perspective of $\advs$, then there exists negligible $\epsilon$ such that $H_{\infty}^\epsilon(K_v)= O(m)$. This implies that if $\tilde{K}_v$ represents the random variable for $\tilde{P}(v)$ from the perspective of $\advs$, then there exists negligible $\epsilon'$ such that $H_{\infty}^{\epsilon'}(\tilde{K}_v)\geq O(3m/4)$. Privacy Amplification Theorem \ref{privacy amplification}, then implies that (whp) $\adv$ cannot distinguish $\tilde{P}_v\coloneqq H(\tilde{P}_{\overline{S}}(v))$ from random. 

$\adv_{\texttt{s}}$ gets a polynomial number of signature and verification oracle queries where verification is performed using $\nk_v\coloneqq (v,\tilde{P}_v)$. Denote the signatures produced by the oracle as $\langle \ot({F}_i),\ot({S}_{P,F_i,\nk_i}), |\overline{c}\rangle_i\rangle_{i\in Q_S}$ and denote the verification queries submitted to the oracle as $\langle \ot (\hat{F}_j),\ot(\hat{S}_j),\hat{c}_j\rangle_{j\in Q_V}$ for $\mathbf{NC}^1$ functions $\hat{F}_j:\{0,1\}^\ell \rightarrow \{0,1\}^{\ell,\ell,\ell}$ and $\hat{S}_j: \{0,1\}^{n,2\ell} \rightarrow  \{0,1\}^{3\ell,\ell,\ell}$. If a query submitted to the verification oracle is not of this form, then it is rejected. 

$\advs$ cannot distinguish the key $\tilde{P}_v$ from random with the signature queries alone since such an attack can be simulated with a single oracle query to each program in the set $\{\ot({S}_{P,F_i,\nk_i})\}_{i\in Q_S}$ by the security of one-time programs (Theorem \ref{thm:otp}). Lemma \ref{Raz 2} shows that $\tilde{P}_v$ remains indistinguishable from random given polynomial guesses.

We now consider what $\advs$ can learn from the first verification query $\langle \ot (\hat{F}_0),\ot(\hat{S}_0),\hat{c}_0\rangle$.
Suppose $\advs$ received $\mathcal{Q}_0\coloneqq \langle \ot({F}_i),\ot({S}_{P,F_i,\nk_i}), |\overline{c}\rangle_i\rangle_{i\in [Q_0]}$ from the signature oracle prior to submitting the verification query. $\advs$ may utilize these signatures in the construction of the verification query. 

Note that, the memory bound for $(\ot(F_i))_{i\in [Q_0]}$ must have already applied before the verification query is answered meaning there exists $(\hat{t}_i)_{i\in [Q_0]}$ such that $\advs$ can only distinguish $F_i(x)$ from random on $x=\hat{t}_i$ for all $i\in [Q_0]$ by the security of one-time programs. 

Now let $\textsf{A}_{0}$ be the sub-algorithm of $\advs$ which produces the verification query using the signatures $\mathcal{Q}_0$. 
When the verification oracle receives the query, it evaluates the first one-time program on a random input say ${t}_0$ to get $\hat{F}_0({t}_0)=(\hat{f}_0(t_0),\hat{f}'_0(t_0),\hat{f}''_0(t_0))$. Next, it evaluates the second one-time program on $\hat{S}_0(v,M_v \cdot (t_0\|\hat{f}_0(t_0))+z_v)$ and gets either an output of the form $(M_v'\cdot (\hat{f}_0(t_0)\|\hat{f}'_0(t_0)\|\hat{f}''_0(t_0))+z_v',\hat{f}'(t_0)\hat{\nk}_0+\hat{f}''(t_0), \hat{\nk}_0)$ for some string $\hat{\nk}_0$ or it rejects and outputs $\perp$. 

The goal is to show that (whp) either the output of the verification query is $\perp$ or that $\advs$ cannot distinguish $\hat{F}_0({t}_0)$ from random. We split the verification procedure performed by the oracle into three parts. Firstly, let $\textsf{Or}_0$ be the sub-algorithm of the oracle which evaluates $\ot(\hat{F}_0)$ on $t_0$ without access to $\tilde{P}_v$. Let $\textsf{Or}_0'$ be the sub-algorithm which evaluates $\ot(\hat{S}_0)$ on $(v,M_v \cdot ({t}_0\|\hat{f}_0({t}_0))+z_v)$ with access only to this input (without access to $\tilde{P}_v$) and outputs the result of the evaluation denoted by $y_0$. Finally, let $\textsf{Or}''_0$ be the sub-algorithm of the oracle with full access to $\tilde{P}_v$ and that checks whether $y_0=(M_v'\cdot (\hat{f}_0(t_0)\|\hat{f}_0'(t_0)\|\hat{f}_0''(t_0))+z_v',\hat{f}'_0(t_0)\hat{\nk}_0+\hat{f}_0''(t_0), \hat{\nk}_0)$ for some key $\hat{\nk}_0$ and correspondingly outputs $\textsf{1QDec}_{\hat{\nk}_0}(|\overline{c}\rangle)$ or $\perp$. We denote this output by $m_0$. Then we can write:
\begin{gather*}
\hat{F}_0(t_0)\leftarrow \textsf{Or}_0\circ \textsf{A}_0(\mathcal{Q}_0) \\
    y_0\leftarrow \textsf{Or}_0'\circ \textsf{A}_0(\mathcal{Q}_0).\\
    m_0 \leftarrow \textsf{Or}''_0(y_0)
\end{gather*} 

Critically, the algorithm $\textsf{A}_0$ is oblivious of $M_v,z_v$ and ${\mathsf{Or}}_0'$ has access to only a single evaluation of the function $Q_v(x)\coloneqq M_v\cdot (x)+z_v$, namely $Q_v({t}_0\|\hat{f}_0({t}_0))$. By the security of one-time programs, the second procedure can be simulated by a simulator $\mathcal{S}_0$ with access to only a single query to each program $(S_{P,F_i,\nk_i})_{i\in [Q_0]}$ and to the value $Q_v({t}_0\|\hat{f}_0({t}_0))$. The query to $S_{P,F_i,\nk_i}$ will yield $\perp$ except with negligible probability unless ${f}_i({t}_0) = \hat{f}_0({t}_0)$ since $\mathcal{S}_0$ cannot guess $Q_v(x)$ on $x\neq {t}_0\|\hat{f}_0(t_0)$. If this condition is not satisfied for any $i\in [Q_0]$, then $\mathcal{S}_0$ will (whp) receive $\perp$ from all its oracle queries. This would mean $\mathcal{S}_0$ cannot guess $Q'_v(x)\coloneqq M_v'\cdot x+z'_v$ on any value and thus the output of $\textsf{Or}''_0$ will be $m_0=\perp$. Note that there is at most a single value $i$ that satisfies this condition (whp) since the parameters of the functions $(F_i)_{i\in [Q_0]}$ are chosen independently at random. 

Assume this condition is satisfied only by the function $F_{n}$ where $n\in [Q_0]$. There is negligible probability that ${t}_0=\hat{t}_{n}$ and so $\advs$ cannot distinguish $F_{n}({t}_0)$ from random. To sum up, (whp) the query to $S_{P,F_i,\nk_i}$ will yield $\perp$ except if $i=n$, in which case $\mathcal{S}_0$ can learn $(Q_v'( {f}_{n}(t_0)\|{f}_{n}'(t_0)\|{f}_{n}''(t_0)),{f}_{n}'(t_0){\nk}_{n}+{f}_{n}''(t_0), {\nk}_{n})$. Thus, the simulator only has a single evaluation of $Q_v'$ and (whp) cannot guess $Q'_v(x)$ on $x\neq {f}_{n}(t_0)\|{f}_{n}'(t_0)\|{f}_{n}''(t_0)$. Hence, we must have either $m_0=\perp$ or $\hat{F}_0(t_0)=F_{n}(t_0)$. In the second case, $\advs$ cannot distinguish $\hat{F}_0(t_0)$ from random which is what we wanted to show. 

Next, we must have (whp) $\hat{\nk}_0=\nk_{n}$ since $\mathcal{S}_0$ cannot guess ${f}_{n}'(t_0)x+{f}_{n}''(t_0)$ on $x\neq \nk_{n}$. By the security and unclonability of one-time programs, since the oracle is able to learn information regarding the evaluation of $S_{P,F_n,\nk_n}$ on $(v,Q_v ({t}_0\|\hat{f}_0({t}_0)))$, the adversary cannot learn any information from the one-time program $\ot(S_{P,F_n,\nk_n})$ and thus cannot authenticate a new message using $\textsf{1QEnc}$ with the key $\nk_n$ by the security of one-time \textsf{QAE} $\Pi_{\textnormal{1QAE}}$.

Overall, the values involved in the verification query are $v$, $Q_v(t_0\|\hat{f}_0(t_0))$, $Q_v' (\hat{f}_0(t_0)\|\hat{f}'_0(t_0)\|\hat{f}''_0(t_0))$, $\hat{f}'(t_0)\hat{\nk}_0+\hat{f}''(t_0)$ and $\hat{\nk}_0$. Even if $\advs$ learns all these values, it does not help in determining the $M_v,M_v',z_v,z_v'$ since $\hat{F}_0({t}_0)$ is indistinguishable from random. 

Notice that the only requirement to apply this argument on a verification query is that $\adv$ cannot distinguish the key $\nk_v$ from random when it submits the query. Hence, this argument can be applied to all verification queries and it can be deduced inductively that $\advs$ cannot distinguish the key from random when it submits the forgery. 

Let $\langle \ot (\hat{F}),\ot(\hat{S}),\hat{c}\rangle$ be the forged signature. By the same arguments used above, if $\hat{\nk}$ is the key output of program $\ot(\hat{S})$ when evaluated by the experiment, then there exists $n'$ such that $\hat{\nk}=\nk_{n'}$ or the signature is invalid. In the former case, we showed that $\advs$ cannot authenticate a new message with $\hat{\nk}$ so either the signature is invalid or ${\hat{c}}\approx_{\negl[\lambda]}\overline{c}$ which is equivalent to the fidelity being close to 1. In either case, the outcome of the experiment is 0 except with negligible probability. 
\qed
\end{proof}

\section{Time-Dependent Quantum Money}

In this section, we present two publicly verifiable quantum money schemes where the money generation and verification algorithms depend on time. The first construction is based on $\TdVK$ signatures with time-dependent verification keys, and the second is based on $\TD$ signatures with time-independent verification key.

\subsection{Definitions}
We first define $\TD$ public-key quantum money. Such a scheme consists of the following algorithms for the bank: $\textsf{KeyGen}$ for generating a secret key, and $\textsf{CoinGen}$ for generating a banknote. For the public user, we define the algorithm $\textsf{CoinVerify}$ which describes how to validate a banknote allowing the public to carry out transactions and payments.

\begin{definition}[$\TD$ Public-Key Quantum Money]
    A \emph{$\TD$ public-key quantum money} scheme consists of the following algorithms for the bank:
    \begin{itemize}
        \item $\textsf{KeyGen}(1^\lambda):$ Outputs a secret key $\sk$ and verification key $\vk$.
        \item $\textsf{CoinGen}(\sk,n,\ttt, v):$ Generates a quantum banknote $\$_v$ worth $v$ dollars with lifespan $n$ based on the time $\ttt$. 
    \end{itemize}
    and the following algorithm for the public user:
    \begin{itemize}
        \item $\textsf{CoinVerify}(\vk,\ttt, v,\$):$ Checks if $\$ $ is a valid banknote worth $v$ dollars using $\vk$ based on the time $\ttt$. Outputs a banknote $\$'$ or $\perp$. 
    \end{itemize}
\end{definition}

We can also allow the verification key to vary with time. This is required for our second construction but not the first. In this case, we use an algorithm $\textsf{vkAnnounce}(\sk,\ttt)$ which describes how to vary the verification key with respect to time. 

\begin{definition}[Correctness]
A $\TD$ public-key quantum money scheme $\Pi$ is \emph{correct} if for any $i,n,v\in\mathbb{N}$, time $T\in [t_i,t_{i+1})$, $(\sk,\vk)\leftarrow \textsf{KeyGen}(1^\lambda)$, and banknote $\$ \leftarrow \textsf{CoinGen}(\sk,n,T,v)$, there is negligible probability (in $\lambda$) that $\$$ gets rejected in $n$ transactions consisting of consecutive applications of $\textsf{CoinVerify}$. More mathematically, correctness requires that:
\begin{align*}
    \|\textsf{CoinVerify}(\vk, T+n, v, &\textsf{CoinVerify}(\vk,T+(n-1),v,... \textsf{CoinVerify}(\vk,T+1,v,\$)...))\\
    &- \perp\|\geq 1-\negl[\lambda].
\end{align*}
\end{definition}

We now define an experiment to test the security of a $\TD$ quantum money scheme. Informally, this experiment requires that an adversary with access to a polynomial number of banknotes cannot generate a set of banknotes of greater value. 

\smallskip \noindent\fbox{%
    \parbox{\textwidth}{%
\textbf{Experiment} $\textsf{TD-PKQM}_{\Pi,\adv}({\lambda},T)$:
\begin{enumerate}
\item Sample $(\sk,\vk)\leftarrow \textsf{KeyGen}(1^\lambda)$.
\item $\adv$ is given $\vk$ and access to an oracle for $\textsf{CoinGen}(\sk,\cdot, \cdot,\cdot )$ which rejects queries for time after $T$. 
\item $\adv$ receives $(\$_{v_i})_{i\in [p]}$ from the oracle, where $v_i$ signifies the value of the banknote.
\item $\adv$ outputs $(c_i,\tilde{\$}_{{c}_i})_{i\in [p']}.$
\item The output of the experiment is 1 if:
\begin{align*}
    \sum_{z: \textsf{CoinVerify}(\vk,T,c_z,\tilde{\$}_{c_z})\neq \perp}c_z> \sum_{j\in [p]}v_j. 
\end{align*}
\item Otherwise, the output is $0$.
\end{enumerate}
}}
    \smallskip  

\begin{definition}[Security]
A $\TdVK$ public-key quantum money scheme $\Pi$ is \emph{secure} if for any QPT adversary $\adv$ and $T\in \mathbb{R}$,
\begin{align*}
    \Pr{[\textsf{TD-PKQM}_{\Pi,{\adv}}({\lambda},T)=1]} &\leq \negl[\lambda].
\end{align*}
\end{definition}

\subsection{Construction for \textsf{TD-DVK}-PKQM}

In this section, we present a scheme for $\TdVK$ public-key quantum money assuming \textsf{pq-PRF}s. Banknotes are quantum states secured with multiple layers of one-time \textsf{QAE}. The bank reveals a classical secret encryption key every second (or some fixed interval $\delta$) allowing the public to validate banknotes by verifying one authentication layer. Eventually, after all the keys used in a banknote's generation have been revealed, the banknote expires. Hence, users must return the banknote to the bank before this occurs for replacement. The scheme is built for money values $v$ such that $v\in \{0,1\}^\lambda$. However, it can easily be generalized to allow for larger values.

\begin{construct}{
\label{con:money}
Let $n,\lambda \in \mathbb{N}$ be security parameters and $p\in \poly[\lambda]$. Let $\Pi_{\textnormal{1QAE}}=(\textsf{1QGen},\textsf{1QEnc},\textsf{1QDec})$ be the algorithms for a one-time \textsf{QAE} scheme on $(n\lambda +1)$-qubit messages such that the encryption of a $m$-qubit message is $(m+\lambda)$-qubits and where the key generation algorithm makes $p$ coin tosses. Let $F:\{0,1\}^\lambda \times \{0,1\}^n\rightarrow \{0,1\}^\lambda$ and $F':\{0,1\}^{\lambda,3\lambda+p} \rightarrow \{0,1\}^{p}$ be \textsf{pq-PRF}s. The construction for a $\TdVK$ quantum public-key money scheme is as follows:
\begin{itemize}
    \item $\textsf{KeyGen}(1^\lambda):$ \textsf{(Bank)} Sample $\sk\leftarrow \{0,1\}^\lambda$. Output $\sk$. 
    \item $\textsf{vkAnnounce}(\sk,\ttt):$ \textsf{(Bank)} If $\ttt=t_{i+1}$ for some $i\in \mathbb{N}$, then compute $\nk_i\coloneqq F(\sk,i)$ and announce $\vk_i\coloneqq (t_i,t_{i+1},\nk_i)$. Otherwise, output $\perp$.
    \item $\textsf{CoinGen}(\sk,n, \ttt,v):$ \textsf{(Bank)} Let $j$ be such that $\ttt \in [t_j,t_{j+1})$. 
    \begin{enumerate}
        \item Sample $y\leftarrow \{0,1\}^{p}$ and set $y_j=y$.
    \item For each $i\in [j:j+n],$ compute:
    \begin{enumerate}        
    \item $\ek_{i}=\textsf{1QGen}(1^\lambda;y_{i})$.
    \item $\nk_i\coloneqq F(\sk,i)$.
    \item $y_{i+1}\coloneqq F'(\nk_{i},i\|v\|t_{j+n}\|y_{i})$.
    \end{enumerate} 
     \item Generate
    \begin{align*}
        \sigma\coloneqq \textsf{1QEnc}_{\ek_{j}}( \textsf{1QEnc}_{\ek_{j+1}}(...\textsf{1QEnc}_{\ek_{j+n}}(\ket{0})...)).
    \end{align*}
    \item Output ${\$}_{v,j,y} \coloneqq (t_j,t_{j+n},v,y_j,\sigma)$.
        \end{enumerate}

    \item $\textsf{CoinVerify}( \vk_i ,v,\$_{v,i,y}):$ 
    \begin{enumerate}
        \item Interpret $\$_{v,i,y}$ as $(t,t',v',y_i,\sigma_i)$ and $\vk_i$ as $(t_i,t_{i+1},\nk_i)$.
        \item If $t\neq t_i$, $t'< t_{i+1}$, or $v'\neq v$, then output $\perp$.
        \item Compute $y_{i+1}\coloneqq F'(\nk_{i},i\|v'\|t'\|y_{i})$.
        \item Compute $\ek_{i+1}=\textsf{1QGen}(1^\lambda;y_{i+1})$.
        \item Compute $\sigma_{i+1} \coloneqq \textsf{1QDec}_{\ek_{i+1}}(\sigma_i)$.
        \item If $\sigma_{i+1}=\perp$, then output $\perp$.
        \item Otherwise, accept banknote and output $\$_{v,i+1,y}\coloneqq  (t_{i+1},t',v,y_{i+1},\sigma_{i+1})$.
    \end{enumerate}
    
\end{itemize}}\end{construct}


\begin{theorem}
    Construction \ref{con:money} is a secure $\TdVK$ public-key quantum money scheme assuming the existence of \textsf{pq-PRF}s.
\end{theorem}

\begin{proof}
We prove security through a sequence of hybrids. The first hybrid is the standard security experiment with $T\in [t_n,t_{n+1})$. 

\smallskip \noindent\fbox{%
    \parbox{\textwidth}{%
\textbf{Hybrid} $\Hy_0$ (standard experiment):
\begin{enumerate}
\item Sample $\sk\leftarrow \textsf{KeyGen}(1^\lambda)$.
\item For each $j\leq n$, generate $\vk_j\leftarrow \textsf{vkAnnounce}(\sk,t_{j+1})$.     
\item $\adv$ is given $(\vk_j)_{j<n}$ and access to an oracle for $\textsf{CoinGen}(\sk,\cdot, \cdot,\cdot )$ which rejects queries for time after $T$. 
\item $\adv$ receives $(\$_{v_i})_{i\in [p]}$ from the oracle, where $v_i$ signifies the value of the banknote.
\item $\adv$ outputs $(c_i,\tilde{\$}_{{c}_i})_{i\in [p']}.$
\item The output of the experiment is 1 if:
\begin{align*}
     \sum_{z:\textsf{CoinVerify}(\vk_n,c_z,\tilde{\$}_{c_z})\neq \perp}c_z> \sum_{j\in [p]}v_j. 
\end{align*}
\item Otherwise, the output is $0$.
\end{enumerate}
}}
    \smallskip 

Hybrid $\Hy_1$ is the same as $\Hy_0$ except the \textsf{pq-PRF} $F(\sk,\cdot)$ is replaced with a completely random function $R(\cdot)$ on the same domain and co-domain. This only changes the $\textsf{CoinGen}$ procedure, where now the keys $\nk_i=R(i)$ are generated using $R$.

Hybrid $\Hy_2$ is the same as $\Hy_1$, except that the \textsf{pq-PRF} $F'(\nk_n,\cdot)$ is replaced with a random function $R'$. 

It is easy to see that hybrids $\Hy_0$ and $\Hy_1$ are indistinguishable by the security of \textsf{pq-PRF}s. Similarly, $\Hy_1$ and $\Hy_2$ are indistinguishable given that $\nk_n$ is indistinguishable from random as it is generated with a random function $R$, allowing us to replace $F'(\nk_n,\cdot)$ with a random function. As a result, it is sufficient to prove security of $\Hy_2$ which is what we do now. 

Assume there exists a QPT adversary $\adv$ such that $\Pr[\Hy_2=1]$ is non-negligible. Assume, without loss of generality, that $1\in \mathcal{C}$ where $\mathcal{C}\coloneqq \{z:\textsf{CoinVerify}(\vk_n,c_z,\tilde{\$}_{c_z})\neq \perp\}$. Interpret $\tilde{\$}_{c_1}$ as $(t_n,\tilde{t}',c_1,\tilde{y}_n,\tilde{\sigma})$. 

When the experiment runs $\textsf{CoinVerify}( \vk_n,c_1,\tilde{\$}_{c_1})$, it computes $\tilde{y}_{n+1}\coloneqq R'(c_1\|\tilde{t}'\|\tilde{y}_{n})$ and then $\ek_{n+1}=\textsf{1QGen}(1^\lambda;\tilde{y}_{n+1})$. If the bank did not perform the evaluation $R'(c_1\|\tilde{t}'\|\tilde{y}_{n})$ during the generation of banknotes, then it is clear that there is negligible probability that $\tilde{\$}_{c_1}$ passes verification. It is also clear that (whp) this evaluation only occurs at most in a single banknote. Without loss of generality, assume ${\$}_{v_1}$ involves this evaluation. 

Interpret $\$_{v_1}$ as $(t,t',v_1,{y},{\sigma})$. We have $R'(n\|c_1\|\tilde{t}'\|\tilde{y}_{n})=R'(n\|{v}_1\|t'\|{y}_{n})$ which implies (whp) that ${c}_1=v_1$, $\tilde{y}_{n}=y_\ell$ and $t'=\tilde{t}'$. 

As $\ek_{n+1}$ is indistinguishable from random with respect to $\adv$, by the security of $\Pi_{\textnormal{1QAE}}$, $\adv$ cannot produce two states that pass verification with key $\ek_{n+1}$ except with negligible probability. In other words, applying this argument to all banknotes, there is an injective mapping from $(c_z)_{z\in \mathcal{C}}$ to $(v_i)_{i\in [p]}$. This means that $\sum_{z\in [\mathcal{C}]}c_z \leq \sum_{j\in [p]}v_j$ giving a contradiction.
\qed
\end{proof}

\remark{We make a short note regarding the efficiency of the scheme. First of all, the one-time \textsf{QAE} scheme $\Pi_{\textnormal{1QAE}}$ utilized in the scheme which encrypts $m$-qubit messages into $(m+\lambda)$-qubit messages exists unconditionally (Lemma \ref{lem:one-time enc}). Now a quantum banknote consists of $n$ applications of the one-time \textsf{QAE} on a single qubit $\ket{0}$. Hence, the banknote is a $(n\lambda+1)$-qubit state, and so grows linearly with respect to its lifespan. }

\subsection{Construction for \textsf{TD}-PKQM}

In this section, we present a public-key quantum money scheme where the generation and verification algorithms depend on time but, importantly, the verification key is time-independent allowing users to be completely offline. We prove security assuming the existence of $\TLP$s and $\OWF$s. 

\begin{construct}{
\label{con:TD money}
Let $n,\lambda \in \mathbb{N}$ be security parameters and $p\in \poly[\lambda]$. Let $\Pi_{\textnormal{1QAE}}=(\textsf{1QGen},\textsf{1QEnc},\textsf{1QDec})$ be the algorithms of a one-time \textsf{QAE} scheme on $(n\lambda+1)$-qubit messages where the key generation algorithm makes $p$ coin tosses. Let $\Pi_{\textnormal{1E}}=(\textsf{1Gen},\textsf{1Enc},\textsf{1Dec})$ be the algorithms of a one-time encryption scheme on classical messages. Let $\Pi_{\textnormal{CS}}\coloneqq ({\textsf{CS.Gen}}, {\textsf{CS.Sign}},{\textsf{CS.Verify}})$ be the algorithms for a signature scheme on classical messages. Let $(\textsf{Puzzle.Gen},\textsf{Puzzle.Sol})$ be the algorithms of a $\TLP$ scheme. The construction for a $\TD$ quantum public-key money scheme is as follows:
\begin{itemize}
    \item $\textsf{KeyGen}(1^\lambda):$ \textsf{(Bank)} Sample $(\sk,\vk)\leftarrow \textsf{CS.Gen}(1^\lambda)$. Output $(\sk,\vk)$.
    \item $\textsf{CoinGen}(\sk,n, \ttt,v):$ \textsf{(Bank)} Let $j$ be such that $\ttt \in [t_j,t_{j+1})$.
    \begin{enumerate} 
    \item Sample tag $r\leftarrow \{0,1\}^\lambda$ and key $\nk_{j-1}\leftarrow \textsf{1Gen}(1^\lambda)$.
    \item For each $i\in [j:j+n]$:
    \begin{enumerate}
        \item Sample $\ek_{i}=\textsf{1QGen}(1^\lambda)$.
        \item Sample $\nk_{i}=\textsf{1Gen}(1^\lambda)$.
        \item Generate $P_i\leftarrow \textsf{Puzzle.Gen}(1^\lambda,1, (\ek_i, \nk_i)).$
        \item Encrypt $Z_i=\textsf{1Enc}_{\nk_{i-1}}(P_i)$.
        \item Sign $s_i \leftarrow \textsf{CS.Sign}(\sk,(Z_i,v,r,t_{j+n},t_i))$. 
    \end{enumerate}
        \item Let $Z\coloneqq (Z_j,\ldots,Z_{j+n})$ and $S\coloneqq (s_j,\ldots, s_{j+n})$.
     \item Generate
    \begin{align*}
        \sigma\coloneqq \textsf{1QEnc}_{\ek_{j}}( \textsf{1QEnc}_{\ek_{j+1}}(...\textsf{1QEnc}_{\ek_{j+n}}(\ket{0})...)).
    \end{align*}
    \item Output ${\$}_{v,r,j} \coloneqq (t_{j+n},r,v,\nk_{j-1},Z,S,\sigma)$.
        \end{enumerate}

    \item $\textsf{CoinVerify}( \vk ,\ttt, v,\$):$ 
    Let $i\in \mathbb{N}$ be such that $\ttt\in [t_i,t_{i+1})$.
    \begin{enumerate}
        \item Interpret $\$$ as $(t_{n},r',v',\nk',Z',S',\sigma')$.
        \item If $v'\neq v$ or $i+1> n$, then output $\perp$.
        \item For each $k\in [n]$, if $\textsf{CS.Verify}(\vk,(Z_k',v,r',t_{n},t_k), S'_k)=\perp$, then output $\perp.$
        \item Compute $P_i\coloneqq \textsf{1Dec}_{\nk'}(Z'_i)$.
        \item Compute $\textsf{Puzzle.Sol}(P_i)$ and interpret the result as $(\ek_{i},\nk_{i})$.
        \item Compute $\sigma_{i} \coloneqq \textsf{1QDec}_{\ek_{i}}(\sigma')$.
        \item If $\sigma_{i+1}=\perp$, then output $\perp$.
        \item Otherwise, accept the banknote and output $\$_{v,r',i+1}\coloneqq  (t_{n},r',v,\nk_{i},Z',S',\sigma_{i})$.
    \end{enumerate}
\end{itemize}}\end{construct}


\begin{theorem}
    Construction \ref{con:TD money} is a secure $\TD$-QM scheme assuming the existence of a EUF classical digital signature scheme and secure $\TLP$s.
\end{theorem}

\begin{proof}
Fix $T\in\mathbb{R}$ and let $m\in \mathbb{N}$ be such that $T\in [t_m,t_{m+1})$. We prove the scheme is secure through a sequence of hybrids. Hybrid $\Hy_0$ is the standard security experiment which we restate here. 

    \smallskip \noindent\fbox{%
    \parbox{\textwidth}{%
\textbf{Hybrid} $\Hy_0$ (standard experiment):
\begin{enumerate}
\item Sample $(\sk,\vk)\leftarrow \textsf{KeyGen}(1^\lambda)$.
\item $\adv$ is given $\vk$ and access to an oracle for $\textsf{CoinGen}(\sk,\cdot, \cdot,\cdot )$ which rejects queries for time after $T$. 
\item $\adv$ receives $(\$_{v_i})_{i\in [p]}$ from the oracle, where $v_i$ signifies the value of the banknote.
\item $\adv$ outputs $(c_i,\tilde{\$}_{{c}_i})_{i\in [p']}.$
\item The output of the experiment is 1 if:
\begin{align*}
    \sum_{z:\textsf{CoinVerify}(\vk,T,c_z,\tilde{\$}_{c_z})\neq \perp }c_z> \sum_{j\in [p]}v_j. 
\end{align*}
\item Otherwise, the output is $0$.
\end{enumerate}
}}
    \smallskip  

Next, hybrid $\Hy_1$ is the same as $\Hy_0$ except that the $\textsf{CoinVerify}$ procedure rejects any coin where the classical part of the banknote was not generated by the bank. To do this, the bank keeps a storage $\textsf{M}$ on information regarding previous banknotes generated. The new algorithms for the bank are as follows:
\begin{itemize}
    \item $\textsf{KeyGen}_{\Hy_1}(1^\lambda):$ Sample $(\sk,\vk)\leftarrow \textsf{CS.Gen}(1^\lambda)$. Output $(\sk,\vk)$.
    \item $\textsf{CoinGen}_{\Hy_1}(\sk,n, \ttt,v):$ Let $j\in \mathbb{N}$ be such that $\ttt \in [t_j,t_{j+1})$.
    \begin{enumerate}
        \item Run $(t_{j+n},r,v,\nk_{j-1},Z,S,\sigma) \leftarrow \textsf{CoinGen}(\sk,n, \ttt,v).$
    \item Store $(Z,v,r,t_{j+n})$ in storage $\mathbf{M}$.
    \item Output $(t_{j+n},r,v,\nk_{j-1},Z,S,\sigma)$.
    \end{enumerate} 

    \item $\textsf{CoinVerify}_{\Hy_1}( \vk ,\ttt, v,\$):$ 
    Let $i\in \mathbb{N}$ be such that $\ttt\in [t_i,t_{i+1})$.
    \begin{enumerate}
        \item Interpret $\$$ as $(t_{j+n},r',v',\nk',Z',S',\sigma')$.
        \item If $(Z',v',r',t_{j+n})\notin \mathbf{M}$, then output $\perp$.
        \item Otherwise, output $\textsf{CoinVerify}( \vk ,\ttt, v,\$).$ 
    \end{enumerate}
\end{itemize}

Next, hybrid $\Hy_2$ is the same as hybrid $\Hy_1$, except we modify the $\TLP$s $P_i$ for $i\geq m$ to instead encrypt $\perp$. The steps written in \textbf{bold} are the difference from $\Hy_1$. The resulting algorithms for the bank are as follows:

\begin{itemize}
    \item $\textsf{KeyGen}_{\Hy_2}(1^\lambda):$ Sample $(\sk,\vk)\leftarrow \textsf{CS.Gen}(1^\lambda)$. Output $(\sk,\vk)$.
    \item $\textsf{CoinGen}_{\Hy_2}(\sk,n, \ttt,v):$ 
    \begin{enumerate}
    \item Let $j$ be such that $\ttt \in [t_j,t_{j+1})$. 
    \item Sample tag $r\leftarrow \{0,1\}^\lambda$ and key $\nk_{j-1}\leftarrow \textsf{1Gen}(1^\lambda)$.
    \item For each $i\in [j:j+n]$:
    \begin{enumerate}
        \item Sample $\ek_{i}=\textsf{1QGen}(1^\lambda)$.
        \item Sample $\nk_{i}=\textsf{1Gen}(1^\lambda)$.
        \item \textbf{If $i\geq m$, generate $P_i\leftarrow \textsf{Puzzle.Gen}(1^\lambda, 1, 0^{\lvert \nk_{i}\rvert +\lvert \ek_{i}\rvert })$.}
        \item Otherwise, $P_i\leftarrow \textsf{Puzzle.Gen}(1^\lambda, 1, (\ek_i, \nk_i)).$
        \item $Z_i=\textsf{1Enc}_{\nk_{i-1}}(P_i)$.
        \item Sign $s_i \leftarrow \textsf{CS.Sign}(\sk,(Z_i,v,r,t_{j+n},t_i))$.
    \end{enumerate}
        \item Let $Z\coloneqq (Z_j,\ldots,Z_{j+n})$ and $S\coloneqq (s_j,\ldots, s_{j+n})$.
    \item \textbf{Store $(Z,v,r,t_{j+n})$ and $(\ek_i,\nk_i)_{i\geq m}$ in storage $\mathbf{M}$.}
     \item Generate
    \begin{align*}
        \sigma\coloneqq \textsf{1QEnc}_{\ek_{j}}( \textsf{1QEnc}_{\ek_{j+1}}(...\textsf{1QEnc}_{\ek_{j+n}}(\ket{0})...)).
    \end{align*}
    \item Output ${\$}_{v,r,j} \coloneqq (t_{j+n},r,v,\nk_{j-1},Z,S,\sigma)$.
        \end{enumerate}

    \item $\textsf{CoinVerify}_{\Hy_2}( \vk ,\ttt, v,\$):$ 
    Let $i\in \mathbb{N}$ be such that $\ttt\in [t_i,t_{i+1})$.
    \begin{enumerate}
        \item Interpret $\$$ as $(t_{j+n},r',v',\nk',Z',S',\sigma')$.
        \item If $v'\neq v$ or $i+1> n$, then output $\perp$.
        \item For each $k\in [n]$, check if $\textsf{CS.Verify}(\vk,(Z_k',v,r',t_{j+n},t_k), S'_k)=\perp$, then output $\perp.$
        \item Compute $P_i\coloneqq \textsf{1Dec}_{\nk'}(Z'_i)$.
        \item Otherwise, compute $\textsf{Puzzle.Sol}(P_i)$ and interpret the result as $(\ek_{i},\nk_{i})$.
        \item Compute $\sigma_{i+1} \coloneqq \textsf{1QDec}_{\ek_{i}}(\sigma')$.
        \item If $\sigma_{i+1}=\perp$, then output $\perp$.
        \item Otherwise, accept the banknote and output $\$_{v,r',i+1}\coloneqq  (t_{j+n},r',v,\nk_{i},Z',S',\sigma_{i+1})$.

             \item Interpret $\$$ as $(t_{n},r',v',\nk',Z',S',\sigma')$.
        \item If $(Z',v',r',t_{j+n})\notin \mathbf{M}$, then output $\perp$.
        \item If $v'\neq v$ or $i+1> n$, then output $\perp$.
        \item For each $k\in [n]$, if $\textsf{CS.Verify}(\vk,(Z_k',v,r',t_{n},t_k), S'_k)=\perp$, then output $\perp.$
        \item Compute $P_i\coloneqq \textsf{1Dec}_{\nk'}(Z'_i)$.
        \item \textbf{If $i\geq m$, then determine $(\ek_{i},\nk_{i})$ from $\mathbf{M}.$}
        \item Otherwise, compute $\textsf{Puzzle.Sol}(P_i)$ and interpret the result as $(\ek_{i},\nk_{i})$.
        \item Compute $\sigma_{i} \coloneqq \textsf{1QDec}_{\ek_{i}}(\sigma')$.
        \item If $\sigma_{i+1}=\perp$, then output $\perp$.
        \item Otherwise, accept the banknote and output $\$_{v,r',i+1}\coloneqq  (t_{n},r',v,\nk_{i},Z',S',\sigma_{i})$.
    \end{enumerate}
\end{itemize}

\begin{claim}
    $\Pr[\Hy_2=1]\leq \negl[\lambda]$. 
\end{claim}
\begin{proof}
Assume that there exists a QPT adversary $\adv$ such that $\Pr[\Hy_2=1]$ is non-negligible. Interpret the banknote $\$_{v_1}$ generated by the bank in $\Hy_2$ as $(t,r,v_1,\nk,Z,S,\sigma)$. Let $\mathcal{C}\coloneqq \{z:\textsf{CoinVerify}(\vk,T,c_z,\tilde{\$}_{c_z})\neq \perp\}$ and assume without loss of generality that $1\in \mathcal{C}$. Consider the corresponding banknote $\tilde{\$}_{c_1}$ submitted by $\adv$. Interpret $\tilde{\$}_{c_1}$ as $(\tilde{t},\tilde{r},c_1,\tilde{\nk},\tilde{Z},\tilde{S},\tilde{\sigma})$. 

By our assumption $(\tilde{Z},c_1,\tilde{r},\tilde{t})\in \mathbf{M}$. Hence, without loss of generality, assume $(\tilde{Z},c_1,\tilde{r},\tilde{t})=(Z,v_1,r,t)$. Since verification is performed at time $T\in [t_m,t_{m+1})$, the experiment retrieves $(\ek_{1,m},\nk_{1,m})$ from storage $\mathbf{M}$ where $(\ek_{1,m},\nk_{1,m})$ correspond to the keys sampled in the generation of $\$_{v_1}$. The experiment performs verification by running $\textsf{1QDec}_{\ek_{1,m}}(\tilde{\sigma})$.

If each forgery with index in $\mathcal{C}$ has a unique tag, then this would imply that there is an injective mapping from the values $(c_i)_{i\in \mathcal{C}}$ to $(v_i)_{i\in [p]}$, giving $\sum_{z\in \mathcal{C}}c_z\leq \sum_{j\in [p]}v_j$ and an experiment outcome of 0. 

Hence, assume without loss of generality that $\adv$ submits another forgery in $\mathcal{C}$ with tag $r$. Then $\textsf{1QDec}_{\ek_{1,m}}$ is performed with the same key $\ek_{1,m}$. If both forgeries pass verification, then this would mean that $\adv$ is given one state, namely $\sigma$, that is authenticated with $\ek_{1,m}$, and produces two states that pass verification with $\ek_{1,m}$. Notice that $\adv$ is not given any other information regarding $\ek_{1,m}$. Hence, $\adv$ can easily be converted into a QPT algorithm that breaks QUF of $\Pi_{1QAE}$ giving a contradiction.
    \qed
\end{proof}

\begin{claim}
No QPT adversary can distinguish between hybrids $\Hy_2$ and $\Hy_1$ with non-negligible advantage. 
\end{claim}

\begin{proof}
The only difference between $\Hy_2$ and $\Hy_1$ is that for $i\geq m$, the $\TLP$ $P_i$ encrypting $(\ek_i,\nk_i)$ is replaced with $\TLP$ encrypting $0^{\lvert \ek_i\rvert +\lvert\nk_i\rvert}$. 

If $\adv$ receives a banknote at time $t_j$, then $\adv$ requires $1$ hour to unravel $P_{j}$ and deduce $(\ek_j,\nk_j)$ by the security of $\TLP$s. Only after this step can $\adv$ decrypt $Z_{j+1}$ to retrieve $P_{j+1}$. Then, $\adv$ again requires 1 hour to decrypt $P_{j+1}$. As $\adv$ submits its forgery at time $T\in [t_m,t_{m+1})$, it cannot distinguish a $\TLP$ $P_i$ hiding $(\ek_i,\nk_i)$ from a $\TLP$ hiding $0^{\lvert \ek_i\rvert +\lvert\nk_i\rvert}$ for any $i\geq m$. 

Specifically, if an adversary can distinguish the two hybrids with non-negligible advantage, then this enables the construction of an algorithm that can distinguish whether $P_i$ is hiding $(\ek_i,\nk_i)$ or $0^{\lvert \ek_i\rvert +\lvert\nk_i\rvert}$ for $i\geq m$ in time less than $t_{m+1}$ contradicting the security of $\TLP$s. 
\qed
\end{proof}

\begin{claim}
No QPT adversary can distinguish between hybrids $\Hy_1$ and $\Hy_0$ with non-negligible advantage. 
\end{claim}

\begin{proof}
The only difference between $\Hy_1$ and $\Hy_0$ is that the experiment rejects any forgeries where the classical part has not been generated by the bank. In particular, if a forgery is interpreted as $(t_{j+n},r',v',\nk',Z',S',\sigma')$, then in $\Hy_1$, we must have $(Z',v',r',t_{j+n})\in \mathbf{M}$. 

However, note that $(Z',v',r',t_{j+n})$ is signed under the scheme $\Pi_{CS}$. Hence, if $\adv$ can distinguish between these hybrids with non-negligible advantage, then, with non-negligible probability, $\adv$ outputs a banknote where $(Z',v',r',t_{j+n})\notin \mathbf{M}$, along with a valid signature of $(Z',v',r',t_{j+n})$. Such an adversary can be converted into a QPT algorithm that breaks the unforgeability of $\Pi_{CS}$ using standard reductions.
\qed
\end{proof}

To sum up, we have deduced that hybrid $\Pr[\Hy_2=1]\leq \negl[\lambda]$, no adversary can distinguish $\Hy_2$ from $\Hy_1$ with non-negligible advantage, and no adversary can distinguish $\Hy_1$ from $\Hy_0$ with non-negligible advantage. Hence, $\Pr[\Hy_0=1]\leq \negl[\lambda]$.
\qed
\end{proof}

\section{Public Key Encryption with Authenticated Quantum Public-Keys}
\label{sec:QPKE}

In this section, we first define security for quantum public key encryption to take into account attacks that tamper with the quantum public-keys. Following this, we present two constructions that satisfy this notion of security by applying our many-time and one-time $\TdVK$ signature schemes to authenticate the public keys.  

\subsection{Definitions}
\label{sec:my def}
We first recall the algorithms for a $\QPKE$ scheme \cite{GC01, MY22, BS23} where the public-key is a quantum state. Due to the unclonability of quantum states, multiple copies of the key must be created and distributed to each user individually -- this is the general approach taken in quantum public-key works \cite{GC01, MY22, BS23}. Thus, we include an algorithm $\textsf{pkSend}$ which generates a public-key copy, and an algorithm $\textsf{ekReceive}$ which describes how to process a public-key copy and extract a classical encryption key. The public-key copies are certified with a certification key that may vary with time -- hence, we include an algorithm $\textsf{ckAnnounce}$ to describe how to announce the certification keys. 

In this work, we only deal with the encryption of classical messages. This is because any public-key encryption scheme on classical messages can be adapted to quantum messages with a simple modification as shown in \cite{AGM21}. 

\begin{definition}[Quantum Public Key Encryption]
A \emph{quantum public encryption scheme} $\Pi$ over classical message space $\hildd{M}$ consists of the following algorithms: 
\begin{itemize}
    \item $\textsf{KeyGen}(1^\lambda):$ Outputs a classical secret key $\sk$. 
    \item $\textsf{pkSend}(\sk,\ttt):$ Output a quantum public-key $\rho$ using the secret key and depending on the time $\ttt$.
    \item $\textsf{ckAnnounce}(\sk,\ttt):$ If $\ttt=t_{i+1}$ for some $i\in \mathbb{N}$, then outputs a classical certification key $\ck_i$ using the secret key $\sk$.
    \item $\textsf{ekReceive}(\rho,\ck_i):$ Extract a classical encryption key $\nk$ from $\rho$ and $\ck_i$. 
    \item $\textsf{Enc}(\nk,\mu):$ Outputs a classical ciphertext $\textsf{c}$ for $\mu \in \hildd{M}$ using $\nk$.
    \item $\textsf{Dec}(\sk, \textsf{c})$: Outputs a message $\mu'$ by decrypting $\textsf{c}$ using the secret key $\sk$. 
\end{itemize}
\end{definition}

\begin{definition}[Correctness]
A $\QPKE$ scheme $\Pi$ is \emph{correct} if for any message $\mu \in \hildd{M}$, $i\in\mathbb{N}$, and time $\ttt\in [t_i,t_{i+1})$,
\begin{align*} \Pr{\left[
\begin{tabular}{c|c}
 \multirow{5}{*}{$\textsf{Dec}(\sk,\textsf{c})=\mu\ $} &   $\sk\ \leftarrow \textsf{KeyGen}(1^\lambda)$ \\ 
  & $\rho\ \leftarrow \textsf{pkSend}(\sk,\ttt)$\\
  &$\ck_i\leftarrow \textsf{ckAnnounce}(\sk,t_{i+1})$\\
 & $\nk\ \leftarrow \textsf{ekReceive}(\rho,\ck_i)$\\
 & $\textsf{c}\ \leftarrow \textsf{Enc}(\nk,\mu)$\\
 \end{tabular}\right]} \geq 1-\negl[\lambda] .
\end{align*}
\end{definition}

We present an experiment to test a notion we call \emph{indistinguishability against adaptive chosen key and ciphertext attack} (IND-qCKCA2) which strengthens the well-known indistinguishability against adaptive chosen ciphertext attack (IND-qCCA2) \cite{BZ133,GSM20}. In our experiment, the adversary can receive a polynomial number of public-key copies and has quantum access to the encryption and decryption oracles (except on the challenge ciphertext) as in the standard experiment. However, in our case, the adversary is additionally given full control over the distribution of the quantum public-keys and can choose the quantum key used for generating the challenge ciphertext. The certification key is assumed to be securely announced and can be used for validating the quantum public-key. 

\smallskip \noindent\fbox{%
    \parbox{\textwidth}{%
\textbf{Experiment} $\textsf{QPKE}^{\text{IND-qCKCA2}}_{\Pi,\adv}({\lambda},j)$:
\begin{enumerate}
    \item Experiment samples classical secret $\sk\leftarrow \textsf{KeyGen}(1^\lambda)$, a bit $b\leftarrow  \{0,1\}$, and creates an empty list $\mathcal{C}$.
    \item For each $i\leq j$, experiment generates $\ck_i\leftarrow \textsf{ckAnnounce}(\sk,t_{i+1})$.
    \item $\tilde{\rho}\leftarrow \adv^{\textsf{pkSend}(\sk)}((\ck_i)_{i< j})$. 
    \item Experiment extracts encryption key $\nk\leftarrow \textsf{ekReceive}(\tilde{\rho},\ck_j)$.
    \item $\adv$ is allowed to make 3 types of queries:
    \begin{enumerate}
        \item \textbf{Challenge queries:} $\adv$ sends two messages $m_0,m_1$ and experiment responds with $\textsf{c} \leftarrow \textsf{Enc}(\nk,m_b)$. Experiment adds $\textsf{c}$ to $\mathcal{C}$.
        \item \textbf{Encryption queries:} Experiment chooses randomness $r$ and encrypt messages in a superposition using $r$ as randomness:
        \begin{align*}
            \sum_{s,m,c}\alpha_{s,m,c} \ket{s}\ket{m,c} \rightarrow \sum_{s,m,c}\alpha_{s,m,c}\ket{s} \ket{m,c \oplus \textsf{Enc}(\nk,m;r)} 
        \end{align*}
        \item \textbf{Decryption queries:} Experiment decrypts all ciphertexts in superposition except the challenge ciphertexts:
        \begin{align*}
            \sum_{s,m,c}\alpha_{s,m,c} \ket{s}\ket{c,m} \rightarrow \sum_{s,m,c}\alpha_{s,c,m} \ket{s}\ket{c,m \oplus f(c)} 
        \end{align*}
        where 
        \begin{align*}
\begin{split}
f(c)  = \begin{cases} 
\textsf{Dec}(\sk,c) & {c}\notin \mathcal{C}\\
\perp & {c}\in \mathcal{C}.\\ 
\end{cases}
\end{split}        \end{align*}
    \end{enumerate}
    \item $\adv$ outputs a guess $b'$.
    \item The output of the experiment is $1$ if $b=b'$ and 0 otherwise. 
    \end{enumerate}}}
    \smallskip  

\begin{definition}[Security]
A $\QPKE$ scheme $\Pi$ is IND-qCKCA2 if for any QPT adversary $\adv$,
\begin{align*}
    \Pr{[\textsf{QPKE}^{\textit{IND-qCKCA2}}_{\Pi,{\adv}}({\lambda})=1]} &\leq \frac{1}{2}+\negl[\lambda].
\end{align*}
\end{definition}

We can also define weaker notions such as IND-qCKPA and IND-qCKCA1 security in a natural way by restricting access to the decryption oracle. We can also weaken the security definition by allowing only classical oracle access to the encryption and decryption functionalities. We denote the resulting notions by IND-CKCA2, IND-CKCA1, and IND-CKPA.

\subsection{Construction with Mixed-State Keys}

We now construct a qCKCA2-secure $\QPKE$ scheme with an unrestricted distribution of mixed-state keys. The quantum public-key is a state encoding an evaluation of a \textsf{pq-PRF} and is signed using a $\TdVK$ signature scheme. Each user can learn a random evaluation of the \textsf{pq-PRF} which can be used as a secret encryption key. Messages are then encrypted under any qCCA2-secure symmetric encryption scheme. Recall that the existence of $\OWF$s implies the existence of a \textsf{pq-PRF}s and a qCCA2-secure symmetric encryption scheme (Lemma \ref{lem:sym}). 

\begin{construct}
\label{con:Asy}
{\small Let $\Pi_{\textnormal{SE}}=(\textsf{SE.Gen},\textsf{SE.Enc},\textsf{SE.Dec})$ be the algorithms for an IND-qCCA2-secure symmetric encryption scheme on $n$-bit messages such that key generation makes $p<\lambda/2$ coin tosses. Let $\textbf{H}_{\lambda,p}$ be a two-universal class of hash functions and let $m$ denote its size. Let $F:\{0,1\}^{\lambda}\times \{0,1\}^n\rightarrow \{0,1\}^\lambda$, $F':\{0,1\}^{\lambda}\times \{0,1\}^\lambda \rightarrow \{0,1\}^{\lambda}$, $F'':\{0,1\}^{\lambda}\times \{0,1\}^n\rightarrow \{+,\times \}^{2\lambda}$ and $F''':\{0,1\}^{\lambda}\times \{0,1\}^n\rightarrow \{0,1\}^{\lg m}$ be \textsf{pq-PRF}s. The construction for $\QPKE$ on $n$-bit messages is as follows:
\begin{itemize}
    \item $\textsf{KeyGen}(1^\lambda):$ Sample $\nk,\nk',\nk'',\nk'''\leftarrow \{0,1\}^\lambda$ and output $\sk\coloneqq (\nk,\nk',\nk'',\nk''')$. 

\item $\textsf{pkSend}(\sk,\ttt):$ Let $i\in \mathbb{N}$ be such that $\ttt \in [t_i,t_{i+1})$.
\begin{enumerate}
    \item Set $\nk_i'\coloneqq F(\nk',i)$, $\nk_i''\coloneqq F(\nk'',i)$ and choose hash function $H_{i}$ from the set $ \textbf{H}_{\lambda,p}$ based on $\nk_i'''\coloneqq F'''(\nk''',i)$. (Note for the final step, $\nk_i'''$ is a string of length $\lg(m)$ and thus can be interpreted to encode a choice of hash function from the set $\textbf{H}_{\lambda,p}$.)
    \item Choose $x\leftarrow \{0,1\}^{n}$ and let $y_x\coloneqq F(\nk,x)$, $s_{x,i}\coloneqq F'(\nk_i',y_x)$, and $\theta_{x,i}\coloneqq F''(\nk_i'',x)$.
    \item Output $\ket{\pk_{i,x}}\coloneqq  |x\rangle \otimes  |y_x\|s_{x,i} \rangle_{\theta_{x,i}}$.
\end{enumerate} 

\item $\textsf{ckAnnounce}(\sk,\ttt):$ If $\ttt=t_{i+1}$ for some $i\in \mathbb{N}$, announce $\ck_i\coloneqq (t_i,t_{i+1},\nk_i',\nk_i'',\nk_i''')$.

\item $\textsf{ekReceive}(\rho, \ck_{j}):$ 
\begin{enumerate}
    \item Interpret $\ck_j$ as $(t_j,t_{j+1},\nk_j',\nk_j'',\nk_j''')$.
    \item Measure the first register of $\rho$ in computational basis to obtain say $v$. 
    \item Measure the second register in basis $F''(\nk_j'',v)$ to obtain say $(y,s)$. 
    \item Determine $H_j$ from $\nk_j'''$.
    \item Compute $\ek_v\coloneqq \textsf{SE.Gen}(1^\lambda;H_j(y_v)).$
    \item If $s=F'(\nk_j',y)$, then output $\nk_v\coloneqq (v,j,\ek_v)$.
    \item Otherwise, output $\perp$.
\end{enumerate}

\item $\textsf{Enc}(\nk_v,m):$ 
\begin{enumerate}
    \item If $\nk_v=\perp$, then output $\perp$.
    \item Interpret $\nk_v$ as $(v,j,\ek_v)$.
    \item Output $(v, j, \textsf{SE.Enc}(\ek_v, m))$.
\end{enumerate}

\item $\textsf{Dec}(\sk,\textsf{ct}):$ 
\begin{enumerate}
    \item Interpret $\textsf{ct}$ as $(v,j,{c})$. 
    \item If ${c}=\perp$, then output $\perp$. 
    \item Otherwise, compute $F(\nk, v)=y_v$ and $F'''(\nk''',j)$ to determine $H_j$.
    \item Compute $\ek_v=\textsf{SE.Gen}(1^\lambda;H_j(y_v)).$
    \item Output $\textsf{SE.Dec}(\ek_v,c)$. 
\end{enumerate} 
\end{itemize}}
\end{construct}

\begin{theorem}
    Construction \ref{con:Asy} is a qCKCA2-secure $\QPKE$ scheme assuming the existence of a qCCA2-secure symmetric encryption and \textsf{pq-PRF}s. 
\end{theorem}

\begin{proof}
We prove security through a sequence of hybrids. Hybrid $\Hy_0$ is the standard experiment adapted to our construction:

\smallskip \noindent\fbox{%
    \parbox{\textwidth}{%
\textbf{Hybrid} $\Hy_0$ (standard experiment):
\begin{enumerate}
    \item Experiment samples classical secret and certification key $\sk\leftarrow \textsf{KeyGen}(1^\lambda)$, a bit $b\leftarrow  \{0,1\}$, and creates an empty list $\mathcal{C}$.
    \item For each $i\leq j$, experiment generates $\ck_i\leftarrow \textsf{ckAnnounce}(\sk,t_{i+1})$.
    \item $\tilde{\rho}\leftarrow \adv^{\textsf{pkSend}(\sk)}((\ck_i)_{i< j})$. Let $(\ket{\pk_{n_i, x_i}})_{i\in p}$ denote the public-key copies received.
    \item Experiment extracts encryption key $\nk\leftarrow \textsf{ekReceive}(\tilde{\rho},\ck_j)$.
    \item $\adv$ is allowed to make 3 types of queries:
    \begin{enumerate}
        \item \textbf{Challenge queries:} $\adv$ sends two messages $m_0,m_1$ and experiment responds with $\textsf{c} \leftarrow \textsf{Enc}(\nk,m_b)$. Experiment adds $\textsf{c}$ to $\mathcal{C}$.
        \item \textbf{Encryption queries:} Experiment chooses randomness $r$ and encrypt messages in a superposition using $r$ as randomness:
        \begin{align*}
            \sum_{s,m,c}\alpha_{s,m,c} \ket{s}\ket{m,c} \rightarrow \sum_{s,m,c}\alpha_{s,m,c}\ket{s} \ket{m,c \oplus \textsf{Enc}(\nk,m;r)} 
        \end{align*}
        \item \textbf{Decryption queries:} Experiment decrypts all ciphertexts in superposition except the challenge ciphertexts:
        \begin{align*}
            \sum_{s,m,c}\alpha_{s,m,c} \ket{s}\ket{c,m} \rightarrow \sum_{s,m,c}\alpha_{s,c,m} \ket{s}\ket{c,m \oplus f(c)} 
        \end{align*}
        where 
        \begin{align*}
\begin{split}
f(c)  = \begin{cases} 
\textsf{Dec}(\sk,c) & \textsf{c}\notin \mathcal{C}\\
\perp & \textsf{c}\in \mathcal{C}.\\ 
\end{cases}
\end{split}        \end{align*}
    \end{enumerate}
    \item $\adv$ outputs a guess $b'$.
    \item The output of the experiment is $1$ if $b=b'$ and 0 otherwise. 
    \end{enumerate}}}
    \smallskip  

\textbf{Hybrid} $\Hy_1$: This is the same as $\Hy_0$ except $F(\nk,\cdot)$, $F(\nk',\cdot)$, $F(\nk'',\cdot)$ and $F'''(\nk''', \cdot)$ are replaced with completely random functions with the same domain and co-domain. 

\textbf{Hybrid} $\Hy_2$: This hybrid is the same as $\Hy_1$ except the functions $F'(\nk_j',\cdot)$ and $F''(\nk_j'',\cdot)$ are replaced with completely random functions $R'$ and $R''$ on the same domain and co-domain.

\begin{claim}
For any QPT adversary, 
\begin{align*}
    \Pr[\Hy_2=1]\leq \negl[\lambda].
\end{align*}
\end{claim}

\begin{proof}
Assume there exists a QPT adversary $\adv$ such that $\Pr[\Hy_2=1]$ is non-negligible. In $\Hy_2$, the adversary gets public-key copies $(\ket{\pk_{n_i, x_i}})_{i\in p}$ and outputs $\tilde{\rho}$. The experiment then extracts a subkey ${\nk}\leftarrow \textsf{ekReceive}(\tilde{\rho},\ck_j)$ as follows. It measures the first register of $\tilde{\rho}$ in the computational basis to obtain an outcome, which we denote as $\tilde{x}$. Then, it measures the second register in basis $R''(\tilde{x})$ to obtain an outcome, which we denote $(\tilde{y},\tilde{s})$. 

The only way that ${\nk}\neq \perp$ is if $\tilde{s}=R'(\tilde{y})$. Since $R'$ is a random function, (whp) this can only occur if one of the key copies $(\ket{\pk_{n_i, x_i}})_{i\in p}$ contains the value $R'(\tilde{y})$. Note also that this must be a key generated after time $t_j$ as $R'$ is only used in this case. 

Assume, without loss of generality, that the first public-key $\ket{\pk_{n_0,x_0}}=|x_0\rangle \otimes |{y}_0\|R'({y}_0)\rangle_{R''(x_0)}$ satisfies $y_0=\tilde{y}$ and $n_0=j$. In this hybrid, $y_0$ is generated uniformly at random. So (whp) $\ket{\pk_{j,x_0}}$ is the only key that contains any information regarding $y_0$. 

We are now in the same setting as the monogamy of entanglement game (see Sec.~\ref{sec:mon}). Recall, in the monogamy of entanglement game, Alice chooses a random binary string and sends it to Bob and Charlie in a random BB84 basis. In this case, Alice is the part of the experiment that produces the public-keys and sends $|{y}_0\rangle_{R''(x_0)}$, Bob is the adversary and Charlie is the part of the experiment that runs $\textsf{ekReceive}(\tilde{\rho},\ck_j)$. Since the second experiment (Charlie) learns $y_0$, monogamy of entanglement implies that $\adv$ must be uncertain about $y_0$. 

Mathematically, if $Y_0$ is the random variable which represents the distribution of $y_0$ from the perspective of $\adv$, then Theorem \ref{thm:mono} implies that $H_{\infty}(Y_0)\geq \frac{\lambda}{5}$. Then, Privacy Amplification Theorem \ref{privacy amplification} says that $H_j(Y_0)$ is indistinguishable from random, given that $H_j$ is sampled using a random function. Therefore, $\adv$ cannot distinguish the key used for encryption from random but can distinguish the bit $b$ from random. As a result, $\adv$ can, with standard reductions, be converted into a QPT adversary that breaks the IND-qCCA2 security of $\Pi_{\textnormal{QAE}}$, giving a contradiction. 
    \qed
\end{proof}

\begin{claim}
No QPT adversary can distinguish between hybrids $\Hy_2$ and $\Hy_1$ with non-negligible advantage.
\end{claim}

\begin{proof}
The only difference between these hybrids is that the functions $F'(\nk_j',\cdot)$ and $F''(\nk_j'',\cdot)$ in $\Hy_1$ are replaced with completely random functions $R'$ and $R''$ in $\Hy_2$. Note that in $\Hy_1$, the keys $\nk_j'$ and $\nk_j''$ are generated using a random function and hence are indistinguishable from random with respect to any adversary. Therefore, by the security of \textsf{pq-PRF}s, no adversary can distinguish between these hybrids.
    \qed
\end{proof}

\begin{claim}
No QPT adversary can distinguish between hybrids $\Hy_1$ and $\Hy_0$ with non-negligible advantage.
\end{claim}

\begin{proof}
The only difference between these hybrids is that $F(\nk,\cdot)$, $F(\nk',\cdot)$, $F(\nk'',\cdot)$ and $F'''(\nk''', \cdot)$ in $\Hy_0$ are replaced with completely random functions in $\Hy_1$. Given that the keys $\nk,\nk',\nk'',\nk'''$ are sampled uniformly at random, no adversary can distinguish between these two hybrids by the security of \textsf{pq-PRF}s.
    \qed
\end{proof}

To sum up, for any QPT adversary, the probability that $\Hy_2$ outputs 1 is negligible, no adversary can distinguish between $\Hy_2$ and $\Hy_1$ with non-negligible advantage, and no adversary can distinguish between $\Hy_1$ and $\Hy_0$ with non-negligible advantage. Therefore, $\Pr[\Hy_0=1]\leq \negl[\lambda]$.
\qed
\end{proof}

Note that the public-keys in Construction \ref{con:Asy} are mixed-state. Other $\QPKE$ constructions also rely on mixed-state keys \cite{KMN23,BS23}, however, it might be desirable to have pure-state keys. Unfortunately, it is not possible to achieve publicly verifiable pure-state public-keys with unrestricted distribution. This is because the impossibility result \cite{AGM21} on standard quantum signatures implies that we cannot sign the quantum public-keys with a time-independent verification key in the standard model. Hence, the verification key must change which implies the signed states must also change. However, this issue can be solved by distributing keys only at the beginning of the protocol as shown in the next section. 

\subsection{Construction with Pure-State Keys}
\label{sec:pure-state}

In this section, we construct a CKCPA-secure $\QPKE$ scheme from PRFSs with pure-state public-keys distributed only at the start of the protocol. The idea is to apply the one-time $\TdVK$ signature scheme (Construction \ref{con:one-time sig}) discussed in Sec.~\ref{sec:one-time sig} to sign the quantum public-keys. Once the certification key is revealed, allowing users to validate the public-keys, new public-keys can no longer be distributed securely so the distribution phase terminates. As the one-time signature scheme is information-theoretically secure, this approach can be applied to any $\QPKE$ scheme with no extra assumptions. 

Our construction is inspired by the $\QPKE$ scheme presented in \cite{GSV23} which is built from PRFSs. Their scheme does not address attacks that tamper with public-keys which is where our contribution is made. Recall, PRFSs are potentially a weaker assumption than $\OWF$s \cite{K21}. 

Unlike the previous construction, users extract a quantum encryption key from the public-key and each encryption requires a new public-key copy. Furthermore, the scheme is defined only for 1-bit messages, however, it can easily be generalized to encrypt a string of bits by encrypting them one by one.

\begin{construct}
\label{con:pure-state}
{\small Let $\{\psi_{k,x}\}_{k,x}$ be a PRFS family with super-logarithmic input-size and let $\Pi_{\textnormal{1QAE}}=(\textsf{1QGen},\textsf{1QEnc},\textsf{1QDec})$ be the algorithms from a one-time \textsf{QAE} scheme. Let $t_{\textnormal{end}}$ be the time that the key distribution ends. The construction for $\QPKE$ is as follows:
\begin{itemize}
    \item $\textsf{KeyGen}(1^\lambda):$ Sample $\nk\leftarrow \{0,1\}^\lambda$ and $\nk'\leftarrow \textsf{1QGen}(1^\lambda)$. Output $\sk\coloneqq (\nk,k')$.
    
    \item $\textsf{pkSend}(\sk,\ttt):$ If $\ttt \leq t_{\textnormal{end}}$, then output $\ket{\pk} \coloneqq \textsf{1QEnc}_{k'}(\frac{1}{\sqrt{2^\lambda}}\sum_{x\in \{0,1\}^\lambda}\ket{x}\ket{\psi_{k,x}})$. Otherwise, output $\perp$.
    
    \item $\textsf{ckAnnounce}(\sk,\ttt):$ If $\ttt=t_{\textnormal{end}}$, then output $\ck\coloneqq (t_{\textnormal{end}},k')$. Otherwise, output $\perp$.

    \item $\textsf{ekReceive}({\rho},\ck, \ttt):$ Interpret $\ck$ as $(t_{\textnormal{end}},k')$. If $\rho$ is received after $t_{\textnormal{end}}$, then output $\perp$. Otherwise, output $\gamma \coloneqq \textsf{1QDec}_{k'}(\rho)$. 

\item $\textsf{Enc}(\gamma,m):$ 
\begin{enumerate}
    \item Measure the first register of $\gamma$ in the computational basis and let $x$ denote the output. Let $\psi$ be the state in the second register.
    \item If $m=0$, then let ${\phi}\coloneqq \psi$ and if $m=1$, then let ${\phi}$ be the maximally mixed state. 
    \item Output $\sigma \coloneqq (x,{\phi})$. 
\end{enumerate} 

\item $\textsf{Dec}(\sk,\sigma):$ 
\begin{enumerate}
    \item Interpret $\sigma$ as $(x,{\phi})$.
    \item If ${\phi}=\ket{\psi_{k,x}}$, then output 0. Otherwise, output 1. 
\end{enumerate} 
\end{itemize}}\end{construct}

\begin{theorem}
    Construction \ref{con:pure-state} is a CKPA-secure $\QPKE$ scheme assuming the existence of PRFS with super-logarithmic input size. 
\end{theorem}

\begin{proof}
    We first introduce the following hybrid for $y\in \{0,1\}^\lambda$:
    
   $\textbf{H}_{y}:$ This is the same as the standard experiment except the public-keys are generated differently. The experiment chooses a random Haar state $\ket{\phi^{\text{H}_y}}$ and sends public-keys of the form $\ket{\pk^{\text{H}_y}}\coloneqq \textsf{1QEnc}_{k'}(\frac{1}{\sqrt{2^\lambda}}\sum_{x\neq y}(\ket{x}\ket{\psi_{k,x}})+\ket{y}\ket{\phi^{\text{H}_y}})$ whenever requested. 
   
   For any $y\in \{0,1\}^\lambda$, the only difference between $\ket{\pk^{\text{H}_y}}$ and $ \ket{\pk}$ is that the state $\ket{\psi_{k,y}}$ in the public-key is replaced with $\ket{\phi^{\text{H}_y}}$. Since $\ket{\psi_{k,y}}$ is indistinguishable from a random Haar state (by the security of PRFS), $\adv$ cannot distinguish between the public-keys in $\textbf{H}_y$ and in the standard experiment. 

   In the standard security experiment, $\adv$ submits a forged key, denoted by $\tilde{\rho}$, which is used by the experiment to generate the challenge ciphertext. By the security one-time authentication scheme $\Pi_{\textnormal{1QAE}}$ (Lemma \ref{lem:one-time enc}), (whp) either $\tilde{\rho}\approx_{\negl[\lambda]} \ket{\pk}$ or the key extracted by the experiment $\textsf{ekReceive}(\tilde{\rho},\ck)$ is $\perp$. Only the first case needs to be studied since in the second case, the output of any encryption is $\perp$ so $\adv$ cannot distinguish between two ciphertexts. 

   Upon receiving $\tilde{\rho}$, the experiment measures the first register to obtain outcome $x'$ and the second register collapses to a state $\tilde{\psi}_{x'}$. Since $\tilde{\rho}\approx_{\negl[\lambda]} \ket{\pk}$, (whp) we have $ {\tilde{\psi}_{x'} }\approx_{\negl[\lambda]} \ket{\psi_{k,x'}}$. Typically $\adv$ cannot distinguish $\ket{\psi_{k,x'}}$ from a random Haar state; however, $\tilde{\psi}_{x'} $ may be entangled with the state stored by $\adv$ which would allow $\adv$ to distinguish it from random. 
   
   Let $R$ be the state stored by $\adv$ when it sends $\tilde{\rho}$ to the experiment and, similarly, let $R^{\text{H}_{x'}}$ denote the state stored by $\adv$ if it was in hybrid $\textbf{H}_{x'}$. Note that $R$ only depends on the adversary and the public-keys since it is produced prior to the encryption of any messages. But $\ket{\pk^{\text{H}_{x'}}} \approx_{\negl[\lambda]} \ket{\pk}$, which implies that $R^{\text{H}_{x'}}\approx_{\negl[\lambda]} R$. As a result, there exists negligible $\epsilon$ such that the state $\ket{\phi^{\text{H}_{x'}}}$ is $\epsilon$-independent of $R^{\text{H}_{x'}}$ as $R^{\text{H}_{x'}}\approx_{\negl[\lambda]} R$ and $R$ is clearly independent of $\ket{\phi^{\text{H}_{x'}}}$.
   
   In hybrid $\textbf{H}_{x'}$, $\adv$ will receive a state close to $\ket{\phi^{\text{H}_{x'}}}$ for the encryption of 0 and a random Haar state for the encryption of 1. Since $\ket{\phi^{\text{H}_{x'}}}$ is $\epsilon$-independent of $R^{\text{H}_{x'}}$, $\adv$ does not learn information regarding $\ket{\phi^{\text{H}_{x'}}}$ from $R^{\text{H}_{x'}}$. Hence, $\adv$ cannot distinguish between $\ket{\phi^{\text{H}_{x'}}}$ and a random Haar state. In other words, $\adv$ cannot distinguish between the encryption of 0 and 1 in $\textbf{H}_{x'}$. As $\adv$ cannot distinguish between the hybrid and the standard experiment, it cannot distinguish the encryptions in the standard experiment as well. 
   \qed
\end{proof}

\printbibliography

\end{document}